\documentclass[11pt]{article}

\usepackage[utf8]{inputenc} % allow utf-8 input
\usepackage[T1]{fontenc}    % use 8-bit T1 fonts
\usepackage{hyperref}       % hyperlinks
\usepackage{url}            % simple URL typesetting      % professional-quality tables
\usepackage{amsfonts}       % blackboard math symbols
\usepackage{nicefrac}       % compact symbols for 1/2, etc.
\usepackage{microtype}      % microtypography

\usepackage{amstext}
\usepackage{amsmath}
\usepackage{amsthm}
\usepackage{amssymb}
\usepackage{bm}
\usepackage{wrapfig}
\usepackage{amsfonts, amsmath}
\usepackage[utf8]{inputenc}
\usepackage{graphicx}
\usepackage{bbm, comment}
\usepackage{subfigure}
\usepackage{color}
\usepackage{float}
\usepackage{wrapfig}
\usepackage{enumitem}
\usepackage{url}
\usepackage{MnSymbol,wasysym,marvosym,ifsym}

\usepackage{tabulary}
\usepackage{booktabs}

\usepackage{natbib}

\usepackage[top=1.in, bottom=1.in, left=1.in, right=1.in]{geometry}
\newcommand{\E}{\mathbb{E}}

\usepackage[resetlabels,labeled]{multibib}
\newcites{Ref}{Other References}

\newtheorem*{theorem*}{Theorem}
\newtheorem{theorem}{Theorem}[paragraph]

\newtheorem{claim}[theorem]{Claim}

\newtheorem{corollary}[theorem]{Corollary}
\newtheorem{example}[theorem]{Example}
\newtheorem*{example*}{Example}
\newtheorem{proposition}[theorem]{Proposition}
\newtheorem{definition}[theorem]{Definition}
\newtheorem{observation}[theorem]{Observation}

\usepackage[ruled]{algorithm2e} % For algorithms

\SetAlFnt{\small}
\SetAlCapFnt{\small}
\SetAlCapNameFnt{\small}
\SetAlCapHSkip{0pt}
\IncMargin{-\parindent}

\usepackage{cleveref}  
\usepackage{thmtools,thm-restate}
\usepackage{tikz}
\usepackage{tikz-cd}

\usepackage{pgfplots}
\usepackage{filecontents}
\usetikzlibrary{pgfplots.groupplots}
\usetikzlibrary{shapes.geometric}
\usepgfplotslibrary{fillbetween}
\selectcolormodel{cmyk}
\usetikzlibrary{shapes,positioning,intersections,quotes}

\title{Peer Expectation in Robust Forecast Aggregation:\\The Possibility/Impossibility}

\author{Yuqing Kong\\ Peking University}
\date{}
\begin{document}

\maketitle

% Abstract. Note that this must come before \maketitle.
\begin{abstract}
Recently a growing literature study a new forecast aggregation setting where each forecaster is additionally asked ``what's your expectation for the average of other forecasters' forecasts?''. However, most theoretic results in this setting focus on the scenarios where the additional second-order information helps optimally aggregate the forecasts. Here we adopt an adversarial approach and follow the robust forecast aggregation framework proposed by Arielia, Babichenkoa, and Smorodinsky 2018. We delicately analyze the possibility/impossibility of the new setting when there are two forecasters that either are refinement-ordered or receive conditionally independent and identically distributed (c.i.i.d.) signals. We also extend the setting to a higher level of expectation setting where we can additionally ask ``what's your expectation for the other forecaster's expectation for ...''. The results show that in the above settings, the additional second-order information can significantly improve the aggregation accuracy, and the higher the order, the higher the improvement. 
\end{abstract}

\section{Introduction}

We would like to know the forecast for tomorrow's weather and receive two conflicting forecasts ``70\% rainy'', ``30\% rainy'' from two reliable forecasters respectively. In general, we may receive conflicting forecasts every day regarding investment, health, and outcome of a paper submission, even if both sources are reliable. 

The forecasters provide different forecasts because they have different private information. However, most times we are not aware of the information structure of the forecasters' information, which leads to a challenge for forecast aggregation. The challenge remains even if we partially know the information structure. For example, in the above running example, even if we know that one of the two forecasters is more informative than the other\footnote{Two forecasters first observe the same signal. Then one forecaster observes an additional signal. }, defined as the refinement ordered setting, we still cannot identify who the more informed forecaster is. 

 %Simply averaging those forecasts may not be optimal.

Arielia, Babichenkoa, and Smorodinsky 2018 propose a framework, the robust forecast aggregation, to analyze forecast aggregation without full knowledge of information structure. In the framework, the aggregator takes the forecasters' forecasts as input (e.g. 70\% rainy, 30\% rainy) and outputs an aggregated forecast (e.g. 50\% rainy). The aggregator may know the family $\mathcal{Q}$ the underlying information structure $Q$ belongs to, but does not know the exact $Q$ the forecasters share. The regret is defined as \[L(\text{aggregator})=\sup_{Q\in \mathcal{Q}} (\text{aggregator's forecast}-\text{optimal forecast})^2 \] where the optimal forecast is defined as the Bayesian posterior for the event ``tomorrow will be rainy'' conditioning on two forecasters' information and $Q$. \citet{arieli2018robust} prove that in the two refinement-ordered\footnote{\citet{arieli2018robust} use the terminology, Blackwell-ordered setting.} forecasters setting, even if the aggregator knows it is the refinement-ordered setting, no aggregator has regret $\leq 0.0225$.

To  circumvent the impossibility result, one can elicit the information structure $Q$ from the forecasters directly based on the idea of implementation theory~\citep{jackson2001crash}. However, it may not be practical to implement in the forecast aggregation setting. \citet{prelec2004bayesian} propose a setting that elicits only partial knowledge about the information structure but is more intuitive. Agents are asked to answer a single multi-choice question (e.g. Will it rain tomorrow? Yes/No) and provide their prediction of the distribution over the options (e.g. I predict 70\% people answer yes, 30\% people answer no). ~\citet{prelec2017solution} aggregate the answers to the surprisingly popular option. \citet{klzhw} focus on the open-response question (e.g. Which state in the United States of America is closest to Africa?) and ask the agent what does she think other people will answer (e.g. I think other people will answer ``Florida''). They rank answers based on a hypothesis ``More sophisticated people know the mind of less sophisticated people, but not vice versa.''~\citep{DBLP:conf/sigecom/KongS18a}. For example, if we receive two conflicting answers ``Florida'' and  ``Maine'' from two agents and both agents predict the other agent will answer ``Florida'', we will pick ``Maine'' as the top-ranking answer. In addition to asking voter's intention “If the election were held today, who would you vote for?”, \citet{rothschild2011forecasting} ask voters' expectations: “Regardless of who you plan to vote for, who do you think will win the upcoming election?” They show that voter expectations lead to more accurate forecasts. 

The above work shows that asking people's predictions about other people's answers is not only practical to implement but also significantly improves the accuracy of the aggregation empirically. The above work focuses on aggregating discrete signals. Recently, a growing literature, including \citet{palley2019extracting,martinie2020using,chen2021wisdom,wilkening2022hidden,palley2022boosting,peker2022extracting}, extend the above setting to forecast aggregation. Each forecaster is additionally asked about her expectation for the average of other forecasters' forecasts. In the setting when there are two forecasters, each forecaster is asked:

\begin{description}
	\item [Forecast] What's your prediction for the probability that tomorrow will be rainy?
	\item [Peer Expectation] What's your expectation for your peer's forecast?
\end{description}

Unlike the original forecast aggregation setting, the above setting has zero regrets in the refinement ordered setting. For example, if the first forecaster forecasts ``70\% rainy'' and expects that the other one forecasts ``30\% rainy'' while the second forecaster forecasts ``30\% rainy'' and expects that the other one forecasts ``30\% rainy'', then ``70\% rainy'' is from a more informative forecaster. 

In other scenarios, recent studies provide various schemes to employ the additional peer expectation\footnote{It is also called meta-prediction/higher ordered belief in the literature.}. Nevertheless, most theoretical results focus on specific information structures or assume the existence of a large number of forecasters, such that the additional information helps optimally aggregate the forecasts. 

We study the benefit of peer expectation more delicately under the framework of robust forecast aggregation. Adopting this framework allows us to 1) compare to the original setting directly; 2) delicately study the amount of improvement even if we cannot optimally aggregate the forecasts with the additional peer expectation; 3) study the limitation of this new setting. We can also compare the natural schemes proposed in previous work under this framework. 

The results show that in various settings, the additional peer expectation significantly improves the aggregation. The higher level of expectations can improve the aggregation more. One important insight is that the peer expectation, as well as the higher level of expectation, reveals more information related to the common prior shared by the forecasters. Therefore, we can use them to either construct a better approximation for the prior or a better estimation for the forecasters' expertise. Both ways can be used to improve the aggregation. We introduce the results in detail in the next section.

\subsection{Summary of Results \& Techniques}

We follow \citet{arieli2018robust} and first focus on the setting where two forecasters are asked to forecast a binary event's outcome $W=0,1$. This setting is sufficiently complicated to provide insightful results. When there are more forecasters, to describe the information structure more succinctly, it is more natural to assume an c.i.i.d. setting where forecasters' identities do not matter. Moreover, it's also reasonable to ask about the peer expectation in this setting because when there are more than two forecasters, it's easier for each forecaster to provide the peer expectation when other forecasters' identities do not matter. Thus, later we will discuss a special setting where there are many c.i.i.d. forecasters.

%, level of expectations, in the forecast aggregation setting, which is similar to the definition of common knowledge proposed in \citet{aumann1976agreeing}

\paragraph{Expectation Hierarchy} To illustrate the results better, we first introduce a recursively defined concept, expectation hierarchy \citep{samet1998iterated} which is similar to belief hierarchy \citep{harsanyi1967games} and common knowledge \citep{aumann1976agreeing}. For each forecaster, level $\ell=1$ expectation is her forecast for the probability that outcome $W=1$, which is also the expectation for $W$. Level $\ell=2$ expectation is her expectation for the other forecasters' forecast\footnote{Level 1 expectation is derived from the first order belief and level 2 expectation is derived from the second order.}. Level $\ell=3$ expectation is her expectation for the other forecasters' expectation for her forecast. We can obtain higher level of expectations by asking ``What is your expectation for your peer's expectation for your forecast? what is your expectation for your peer's expectation for your expectation for your peer's forecast? '' and so forth. The original forecast aggregation is interpreted as the level $1$ setting. This paper's focus is the peer expectation setting $\ell=2$, the most practical setting with nontrivial level of expectations.

 We study the possibility/impossibility of peer expectation in various natural scenarios studied by \citet{arieli2018robust}. 

\paragraph{Refinement-ordered} As we explained before, the most practical peer expectation setting ($\ell=2$) implies zero regret, while the original setting $\ell=1$ has regret at least $0.0225$ \citep{arieli2018robust}. 

% which strictly dominates\footnote{An aggregator $a$ strictly dominates another $a'$ if for all $Q$, $a$'s regret is strictly lower than $a'$'s unless $a'$'s regret is zero for $Q$. } the average prior scheme, $a_{avg}$, proposed in \citep{arieli2018robust}\footnote{\citet{arieli2018robust} denote the aggregator by $f_{avg}$. We change the notation because we use $a$ to represent aggregator in this paper. }.

\paragraph{Conditionally Independent} Here two forecasters' information are independent conditioning on the outcome. A special case is the setting where the forecasters' signals are not only conditionally independent but also identically distributed (C.I.I.D.). 
	\begin{description}
	\item [Upper Bound (C.I.I.D.)] We propose a simple aggregator, C.I.I.D. aggregator $a_{ciid}$ (\Cref{def:iid}). We show that $a_{ciid}$'s regret is the maximum of a continuous function of 5 variables, where the numerical estimation is 0.00391*. We then improve the above results by providing a modified aggregator, the Hard Sigmoid aggregator $a_{hs}$ (\Cref{def:hs}), that is based on a hard sigmoid function, whose regret is also the maximum of a function of 5 variables, where the numerical estimation is 0.00211*.
	\item [Upper Bound] We reduce the computation of any aggregator's regret to computing the maximum of a function of 9 variables. We propose a simple aggregator, average expectation aggregator $a_{ae}$ (\Cref{def:ae}). $a_{ae}$ uses the average of the peer expectations to estimate the prior.  The numerical estimation for $a_{ae}$'s regret is 0.0072*\footnote{The star mark here and in the table shows that the maximum is numerically estimated by Matlab.}. The reduction result also allows us to compare with other natural aggregators that work for two forecasters proposed in other works (see \Cref{table:compare}). We then improve the above results by providing a modified aggregator, $a_{we}$ (\Cref{def:we}). $a_{we}$ uses the weighted average of the peer expectations to estimate the prior. The regret estimation of $a_{we}$ is 0.0040*. Finally, we combine the hard sigmoid trick proposed in the c.i.i.d. setting and propose the Weighted Hard Sigmoid aggregator $a_{whs}$ (\Cref{def:whs}) whose regret estimation is 0.00255*.
	\item [Lower Bound] We prove that no aggregator can have regret $\geq 0.00144$ in this setting (including the special c.i.i.d. setting). 
	\end{description}
	
\paragraph{Conditionally Independent Conditioning on a Shared Signal}  A slight generalization of the conditionally independent setting is the setting where the forecasters observe a common signal $S$ and their private signals $S_1,S_2$ are independent conditioning on $S$ and the outcome $W$. We show that this generalized setting has the same upper/lower bound as the conditionally independent setting. 
	
\paragraph{Many C.I.I.D.} When many forecasters observe identically distributed signals  which are independent conditioning on the outcome, each forecaster will be asked ``what's your expectation for the average of other forecasters' forecast?'' \citet{arieli2018robust} implicitly show that for all aggregator, there must exists c.i.i.d. $Q$ such that when the number of forecasters goes to infinite, the aggregator's regret $\geq 0.093$. In contrast, when we query their expectations for a random forecaster's forecast, we can easily construct an aggregator $a^n_{ciid}$ whose regret goes to zero when the number of forecasters $n$ goes to infinite \footnote{\label{manyiid}\citet{chen2021wisdom} also show this result with a slightly different analysis. They assume infinite forecasters initially while we write down the regret of finite n and shows that the regret converges to zero.}.

\paragraph{Higher Levels} In the conditionally independent setting, we can construct an aggregator whose regret goes to zero as the level $\ell$ goes to infinite. In general, we show that in a large family of information structures, including the conditionally independent setting, the level $\ell$ expectation converges to the two forecasters' prior forecast. In general, there exists counter examples that the regret is non-trivial even with infinite level of expectations.

%\yk{change the writing, highlight the key results}

%Therefore, any aggregator that uses prior forecast (\yk{add citations}) can be adopted here by using the high level expectations as estimation of prior. 

%\citet{arieli2018robust} show that with only forecasts $\ell=1$, 1) no aggregator has regret $\leq 0.0225$ and 2) the computation of any aggregator's regret can be reduced to maximizing a 

%In this setting, if we are given the prior forecast $\mu$, we can optimally aggregate the forecasts. Nevertheless, we do not have $\mu$. \citet{arieli2018robust} use the average of forecasts to estimate $\mu$ to propose  

%We show that for any aggregator, the worst $Q$ that leads to the maximal regret must be $Q$ where each forecaster's private signal's support size $\leq 3$.

\begin{table*}

	\begin{center}
\begin{tabular}{ c|c|c|c|c } 
 \hline
 & 2-Refinement & 2-Conditionally & 2-C.I.I.D. & $\infty$-C.I.I.D. \\
 & Ordered & Independent &  & (Fixed Q) \\
 \hline \hline  
 Original [ABS 2018] & 0.0225 & (0.0225,0.0250*) & (0.0225,0.0250*) & $\exists\ Q$, (0.093,-)  \\
($\ell =1$) & &  &   \\

 \hline 
 + Peer Expectation  & 0 & (0.00144,0.00255*) & (0.00144,0.00211*) & $0 ^{\ref{manyiid}}$  \\
($\ell =2$) &  &  &  &  \\
 \hline 
 + Exp of exp of ... & 0 & 0 & 0 & 0 \\
 ($\ell =\infty$) &  &  &  &  \\
 
 \hline
\end{tabular}
\end{center}
\caption{Regret Bounds: The first row is the original setting where the forecasters only report forecasts, the second row is the setting where the forecasters additionally report expectations for other forecasters' forecasts, the third row is the setting where the forecasters report an infinite level of expectations. Different columns indicate different information structures among the forecasters. (LB,UB) means that in the setting, no aggregator can have regret $<LB$, and there exists an aggregator whose regret is $\leq UB$.}
\end{table*}

\paragraph{Insights Behind the Aggregators} To analyze the conditionally independent setting, first note that with the prior forecast $\mu$, we can optimally aggregate the forecasts. Thus, constructing the aggregator is reduced to constructing a proxy for the unknown prior $\mu$. To design the proxy, we have the following insights. \begin{description}
	\item [Higher Level $\Rightarrow$ Better Proxy for Prior] In the original setting, \citet{arieli2018robust} use the average of forecasts to estimate $\mu$. In the new setting, we use the average of peer expectations to estimate $\mu$. We show that the peer expectation ($\ell=2$) is a better proxy for $\mu$ than the forecast ($\ell=1$). 
	As level $\ell$ goes to infinite, the estimation converges to $\mu$ thus the regret goes to zero. 
	\item [Distinct Forecasts $\Rightarrow$ Perfect Proxy (C.I.I.D.)] In the c.i.i.d. setting, when the forecaster reports distinct forecasts, we can recover $\mu$ directly with the forecasts and the additional peer expectations. This insight significantly eases the upper bound analysis in the c.i.i.d. setting, while the non-trivial upperbound analysis\footnote{A trivial upperbound analysis for c.i.i.d. directly uses the upperbound for the general two conditionally independent forecasters setting as the upperbound for c.i.i.d.} for c.i.i.d. remains to be an open question in the original setting. 
	\item[Same Forecast \& Expectation $\Rightarrow$ Forecast = Prior] When forecaster 1's forecast and peer expectation are the same, either the other forecaster is a perfect forecaster who knows $W$, or forecaster 1's forecast is the prior. In another case, when forecaster 1's peer expectation is far from her forecast, the unknown prior $\mu$ will be even further. This insight is more delicate compared to the above two. It inspires us to use weighted average to estimate the prior, combined with a hard sigmoid function. This significantly reduces the regret.%	\item [Higher $|f-p|$, Further $\mu$] In the c.i.i.d. setting, the above insight shows that we can design aggregators such that the regret only happens when the forecasters report the same forecast $f$ and peer expectation $p$. We will show that if $f=p$, then $f=p=\mu$. That is, when we receive similar $f$ and $p$, we can be confident use $f$ or $p$ as a proxy for $\mu$, otherwise, $\mu$ will also be further from $f$ and $p$ 
\end{description}

%Improving with Weights + Hard Sigmoid

\paragraph{Main Techniques} We highlight some of the techniques especially the ones that are different from the original setting. 
\begin{description}
    \item[Relationship between Forecasts \& Expectations] The original setting only deals with forecasts. The forecast only depends on the relationship between her private signal and $W$. In contrast, expectation depends on the relationship among three random variables $S_1,S_2,W$, which may involve many factors especially when the signal's support size $|\Sigma|$ is large. 
One useful observation is that in the conditionally independent setting, forecaster $1$'s expectation for forecaster 2's forecast can be written as $f_1 o_2 + (1-f_1) z_2$ where $f_1$ is forecaster $1$'s forecast and $o_2,z_2$ describe the expectation for forecaster 2's forecast from the perspective of a third party who knows the outcome $W$. This simple observation simplifies the analysis.
	\item [Bilinear Programming] In the conditionally independent setting, aided by the above observation, we can reduce the regret computation to a bilinear programming with 3 linear constraints for each side. This helps reduce the size of $Q$'s support in the regret computation. In the c.i.i.d. setting, combined with the second insight (\emph{distinct forecasts...}), we can also reduce the regret computation of a large family of aggregators, while the non-trivial reduction for c.i.i.d. remains to be an open question in the original setting.
	\item [Improving with Weighting \& Hard Sigmoid] In the original setting, the natural forecasts average proxy for the prior has already achieved relatively good performance ($0.0260$) compared to the lowerbound ($0.0225$). Unlike the original setting, in the new setting, there exists a certain amount of gap between either the forecasts average proxy ($0.0260$) or the expectations average proxy ($0.0072$) and the lowerbound ($0.00144$). The major reason is that directly using forecasts or expectations average ignores the relationship between forecasts and expectations. We develop more sophisticated techniques, including weighting and hard sigmoid functions, to utilize the forecasts and expectations simultaneously. 
\end{description}

\subsection{Related Work}

Starting from \citet{bates1969combination}, a large literature studies forecast aggregation. Many consider empirical validation of the schemes \citep{makridakis1982accuracy}, a setting with historical data \citep{10.2307/2344546}, a repeated setting \citep{vovk1990aggregating}, a setting the aggregator knows the information structure \citep{bunn1975bayesian}. See survey of many other works in \citet{clemen1989combining, timmermann2006forecast}. 

\paragraph{One-shot Robust Forecast Aggregation} This work's setting focuses on the setting of one-shot robust forecast aggregation \citep{arieli2018robust,neyman2022you,LEVY2022105075}. \citet{arieli2018robust} propose an analysis framework of robust forecast aggregation where the aggregator aims to aggregate the forecasts for a binary event's outcome, with very limited knowledge about the information structure. \citet{neyman2022you} consider the setting where the outcome can be continuous and propose a family of information structure where aggregating by simply averaging forecasts is sufficiently good. \citet{LEVY2022105075} focus on the setting where the aggregator knows the marginal prior distributions of the forecasters and wants to design scheme that is robust to the correlation structure among the forecasters. \citet{de2021robust} explore a similar setting. This work's settings follow from the setting of \citet{arieli2018robust}. %In this setting, 

% where the forecast is for a binary event's outcome. In this setting, and aggregating by averaging forecasts is suboptimal in various scenarios. In the setting where the forecasters' information are independent conditioning on the outcome, there is no need to analyze the benefit of knowing the prior because it will lead to the optimal aggregation trivially. 

In contrast to the above work, we analyze a new setting where we have access to each forecaster' expectation for the other forecaster' forecast. The next paragraph shows that this setting has drawn increasing attention recently. 

\paragraph{One-shot Information Aggregation with Second Order Belief} \citet{prelec2004bayesian} proposes the framework where agents are asked to provide both their answers and predictions for a single multi-choice question. \citet{prelec2017solution} aggregate people's answers to the option which is surprisingly popular compared to the prior constructed by the predictions. \citet{klzhw} focus on open-response question and ask agents what do they think other people will answer. They use the information to rank the answers without any prior knowledge. \citet{hosseini2021surprisingly} and \citet{schoenebeck2021wisdom} employ this framework to rank a set of predetermined candidates (e.g. which painting is better? which presidential candidate is better?). To forecast the election outcome, \citet{rothschild2011forecasting} show that voters' expectations for other people's votes provide a more accurate forecast. 

The above work focuses on aggregating discrete signals. A growing literature, including \citet{palley2019extracting,martinie2020using,chen2021wisdom,wilkening2022hidden,palley2022boosting,peker2022extracting}, considers the one-shot forecast aggregation with a second order belief. 

This work differs from the prior work in two aspects. First, most prior work \citep{palley2019extracting,martinie2020using,wilkening2022hidden,palley2022boosting,peker2022extracting} focus on the setting where the goal is to estimate a continuous outcome, e.g. a coin's bias. Ideally, each agent can flip the coin several times independently and the optimal aggregated estimation is the simple average of their individual estimations. However, agents share a certain amount of draws. \citet{palley2019extracting} show that with partial knowledge about the relationship between the shared information and individual information, the additional second-order belief recovers the shared information. The optimal aggregation is still linear with a calibration for the shared information. We focus on another natural setting where agents' private signals are independent conditioning on the outcome. In contrast, here the optimal aggregation is non-linear. We pick this setting such that we can compare it to the original setting \citep{arieli2018robust} directly. \citet{chen2021wisdom} also consider this setting while they assume an infinite number of agents where the second-order belief helps optimally aggregate the forecasts. 

Secondly and more importantly, we analyze the benefit of the higher level information by extending the adversarial framework in \citet{arieli2018robust}. This allows us to delicately quantify the possibility/impossibility of the second-order belief, even when the second-order belief cannot help optimally aggregate the forecasts. 

\paragraph{Higher Order Belief} \citet{chen2021wisdom} consider the setting where the higher order belief can be elicited. In this case, each forecaster is asked to predict the joint distribution over their peers' private signals and outcome, such that the perfect aggregation is possible with the full hierarchy. We consider the setting where the forecasters only need to report their higher-level expectations \cite{samet1998iterated}, which can be seen as a compression of their higher-order beliefs. This setting is more simplified, while there exist situations where the perfect aggregation is not possible even with the full hierarchy. For example, when two agents' signals are independent bits where each bit is 0 with probability 0.5 and $W$ is their xor, the forecasts and all levels of expectations are 1/2. Thus, in this case, the perfect aggregation is not possible even with the full hierarchy of expectations. 

Belief hierarchies \citep{harsanyi1967games} play an important role in incomplete information game theory. A similar concept is also used in the notion of common knowledge \citep{aumann1976agreeing}. The higher-order belief can provide better aggregation by revealing the correlation between the agents' private signals. \citet{wang2021forecast} empirically show that estimating the agents' expertise according to their forecasts' correlation can improve the results of forecasts aggregation consistently. Hopefully, the one-shot information aggregation with higher-order belief can provide more scenarios for the application of belief hierarchies.

\paragraph{Prior-Independent Mechanism Design \& Implementation Theory} Robust forecast aggregation shares the same spirit as the prior-independent mechanism design~\citep{dhangwatnotai2010revenue,devanur2011prior,chawla2013prior,fu2015randomization,allouah2020prior, hartline2020benchmark}. The problem is formulated as a min-max game between the designer, who designs the aggregator/mechanism, and nature, who picks the prior distribution from a class. The designer does not have access to the prior distribution but knows the class. Without prior distribution, implementation theory~\citep{jackson2001crash} elicits the prior distribution from agents who participate in the mechanism. This work focuses on eliciting second-order information, i.e., agents' prediction for other agents. As we mentioned before, this setting has been implemented empirically and proved to be practical by prior studies. Nevertheless, the second-order information only reveals partial information about the prior distribution. This work analyzes the possibility/impossibility of the additional second-order information delicately in a natural setting.

 % 

%We hope this work can start a new direction: one-shot forecast aggregation with a higher level information. 

%One may ask why not query the information structure or prior forecast from the forecasters directly. Note that asking for prediction about others is more natural to implement in practice. When the forecasters are strategic, we can incentivize truthful reports by paying them based on proper scoring rules after the outcome is revealed. 

% show that in this case, no aggregator can 

\section{Model}

\paragraph{Forecasts} Two agents\footnote{We also call them forecasters.} are asked to forecast tomorrow's weather $W$ whose outcome is $0$ (sunny) or $1$ (rainy). Each agent $i=1,2$ receives a private signal $s_i\in \Sigma$, which is the realization of random variable $S_i$. Agents share the same prior $Q$ which is the joint distribution over agents' private signals $(S_1,S_2)$ and tomorrow's weather $W$. Agent $i$ who receives $s_i$ will report her forecast $f_i=\Pr_Q[W=1|S_i=s_i]$. We assume $\mu=\Pr_Q[W=1]\in(0,1)$ because otherwise both agents will report 0 or 1 and there is no need for forecast aggregation. 

\paragraph{Peer Expectation} We also ask each agent $i$ to report her expectation for the other agent's forecast, $p_i$. $p_1 = \E_Q[f_2|S_1=s_1]=\sum_x \Pr_Q[S_2=x|S_1=s_1] \Pr_Q[W=1|S_2 = x]$. $p_2$ is defined analogously. 

\paragraph{Aggregator} An aggregator $a(\cdot)$ takes $f_1,f_2,p_1,p_2$ as input and outputs a forecast $a(f_1,f_2,p_1,p_2)\in [0,1]$ for the probability that tomorrow will rain $W=1$. Note that the aggregator does not have the knowledge of $Q$. The naive aggregator is the average of forecasts $a_{naive}(f_1,f_2,p_1,p_2)=\frac{f_1+f_2}{2}$. 

Given agents' private signals $s_1,s_2$ explicitly and with the knowledge $Q$, the optimal forecast is $f^*=\Pr_Q[W=1|S_1=s_1,S_2=s_2]$. We care about the regret compared to the optimal forecast

\[L(a)=\sup_Q L(a,Q) = \sup_Q \E_Q(a(f_1,f_2,p_1,p_2)-f^*)^2.\]

The problem will be $\min_a L(a)$. 

\paragraph{Strategic Agents} In the above model, agents are non-strategic and report their forecasts and expectations truthfully. However, even if agents are strategic, if we can pay them later when $W$ is revealed, we can easily incentivize truthful reports for both forecasts and peer expectations with proper scoring rules (\Cref{sec:strategic}).

\subsection{Warm-up: Two Refinement-Ordered Agents}

Two agents are refinement-ordered if they both observe the same signal $s$ which is a realization of r.v. $S$, and one agent, say agent 1, observes an additional signal $s'$ which is a realization of r.v. $S'$. Formally, $S_1=(S,S')$ and $S_2=S$. However, the aggregator does not know which agent is more informative. In this case, we use the following aggregation scheme:
\[ a(f_1,f_2,p_1,p_2) = \begin{cases}
	f_1 & \text{if $p_1=f_2$}\\
	f_2 & \text{otherwise}
\end{cases} \]

\begin{observation}\label{obs:blackwell}
When two agents are refinement-ordered, they will have the same peer expectations $p_1=p_2$, and the above aggregator has zero regret.\end{observation}

\begin{proof}[Proof of \Cref{obs:blackwell}]
	We first consider the case when agent 1 is the more informative agent. When agent 1 receives $(s,s')$, she knows agent 2's private signal $S_2=S=s$, her expectation for agent 2's forecast is $p_1=f_2$. For agent 2, her expectation for agent $1$'s forecast is $p_2=\sum_x \Pr_Q[S'=x|S=s]\Pr_Q[W=1|S'=x,S=s]=\sum_x \Pr_Q[W=1,S'=x|S=s]=\Pr_Q[W=1|S=s]=f_2$. Thus, $p_1=p_2=f_2$, $a(f_1,f_2,p_1,p_2)=f_1$. 

When agent 2 is more informative, $p_1=p_2=f_1$, if $p_1=f_2$, then $f_1=f_2$ and $a(f_1,f_2,p_1,p_2)=f_1=f_2$, otherwise, $a(f_1,f_2,p_1,p_2)=f_2$ as well. Therefore, the aggregator always outputs the forecast of the more informative agent. 
\end{proof}

%\paragraph{A slightly relaxed version} Two agents are ordered if one agent, say agent 1, owns a better signal in the sense that $S_2$ is independent of $W$ conditioning $S_1$. In this case, we use the following aggregation scheme:
%\[ a(f_1,f_2,p_1,p_2) = \begin{cases}	f_1 & \text{if $f_2=p_2$}\\	f_2 & \text{otherwise}\end{cases} \]

%\begin{observation}\label{obs:order}When two agents are ordered, the above aggregator has zero regret.\end{observation}

%\begin{proof}[Proof of \Cref{obs:order}]
%	We consider the case when agent 1 is the more informative agent. Agent 2's expectation for agent $1$'s forecast is $p_2=\sum_x \Pr_Q[S_1=x|S_2=s_2]\Pr_Q[W=1|S_1=x]=\sum_x \Pr_Q[S_1=x|S_2=s_2]\Pr_Q[W=1|S_1=x,S_2=s_2]=\sum_x \Pr_Q[W=1,S_1=x|S_2=s_2]=\Pr_Q[W=1|S_2=s_2]=f_2$. Thus, $p_2=f_2$, $a(f_1,f_2,p_1,p_2)=f_1$. 

%When agent 2 is more informative, $p_1=f_1$. If $p_2\neq f_2$, $a(f_1,f_2,p_1,p_2)=f_2$. If $p_2=f_2$, ...

%Therefore, the aggregator always outputs the forecast of the more informative agent. 
%\end{proof}

%\yk{approximated version}

\section{Two Conditionally Independent Agents}

In this setting, $S_1$ and $S_2$ are independent conditioning on $W$. 

\subsection{A Useful Observation: Taking $W$'s Perspective}\label{sec:obs}

This section introduces a useful observation that will be used multiple times. Each agent's forecast only depends on the relationship between her private signal and $W$. However, her expectation depends on the joint distribution among three random variables $S_1,S_2,W$, which involve many parameters especially when the signal's support size $|\Sigma|$ is large. 

In contrast, in the conditionally independent setting where agents' private signals are independent conditioning on $W$, the peer expectation can be described succinctly. 

In detail, it's useful to have an imagined third party who knows the outcome $W$. It's sufficient to describe each agent's expectation using her forecast and the third party's expectations. 

\paragraph{Notations $\mu, f_i^x, q_i^x, \mathbf{f}_i,\mathbf{q}_i$} Given $Q$, let $\mu$ denote $\Pr_Q[W=1]$. For all $x\in\Sigma$, let $f_i^x$ denote agent i's forecast conditioning on she receives signal $x$, $\Pr_Q[W=1|S_i=x]$. Let $q_i^x$ denote the prior probability that agent $i$ receives $x$, i.e., $\Pr_Q[S_i=x]$. For $i=1,2$, let $\mathbf{f}_i=(f_i^x)_{x\in\Sigma}$ denote a $|\Sigma|$-dimensional \emph{forecast vector}. Let $\mathbf{q}_i=(q_i^x)_{x\in\Sigma}$ denote a $|\Sigma|$-dimensional \emph{prior vector}. %In the c.i.i.d. setting, because agents' identities does not matter, we omit the subscript $i$ and use $f^x, q^x, \mathbf{f}, \mathbf{q}$. 

\begin{definition}[$W$'s perspective $o_i,z_i$]
	From the perspective of a third party who knows the outcome $W$, let $z_i$ (\textbf{z}ero) denote his expectation for agent $i$'s forecast conditioning on $W=0$, $\sum_x f_i^x \Pr_Q[S_i=x|W=0]$ and $o_i$ (\textbf{o}ne) denote his expectation for agent $i$'s forecast conditioning on $W=1$, $\sum_x f_i^x \Pr_Q[S_i=x|W=1]$ for $i=1,2$. \end{definition}

\Cref{claim:fix} shows that an agent 1's expectation for agent 2 is a convex combination of the third party's expectations, where the convex coefficients are agent 1's forecast $f_1,1-f_1$. 

\begin{restatable}{observation}{claimfix}\label{claim:fix}
$p_1=(1-f_1) z_2 + f_1 o_2$, $p_2= f_2 o_1 + (1-f_2) z_1$, $\mu o_i + (1-\mu) z_i=\mu, i=1,2$ and the optimal forecast is $f^*=\frac{(1-\mu) f_1 f_2}{(1-\mu)f_1 f_2 + \mu(1-f_1)(1-f_2)}$. 
\end{restatable}

Proof of \Cref{claim:fix} is deferred to \Cref{sec:proofs}. The above succinct form allows us to reduce the support size later in the upper-bound analysis.

\subsection{Upperbound}

\paragraph{Average Expectation Aggregator}

In the conditionally independent case, if we are given $\mu=\Pr_Q[W=1]$, we can optimally aggregate the agents' forecast to $\frac{(1-\mu) f_1 f_2}{(1-\mu)f_1 f_2 + \mu(1-f_1)(1-f_2)}$ \citep{bordley1982multiplicative}. However, we do not know $\mu$. \citet{arieli2018robust} use the average of forecasts to estimate $\mu$. Let $a_{avg}$ denote this aggregator. We use the average of the expectation of forecasts as an estimation for $\mu$. 

\begin{definition}[Average Expectation Aggregator] \label{def:ae}
\[ a_{ae}(f_1,f_2,p_1,p_2) = \frac{(1-\hat{\mu}) f_1 f_2}{(1-\hat{\mu})f_1 f_2 + \hat{\mu}(1-f_1)(1-f_2)}\]
where $\hat{\mu} = \frac{p_1+p_2}{2}$. 

When one of $f_1,f_2$ is 0 (1), the aggregator will output 0 (1)\footnote{So do the following aggregators. We will omit this part in the descriptions of the following aggregators.}. 

\end{definition}

In $a_{avg}, \hat{\mu}=\frac{f_1+f_2}{2}$ and its regret estimation is 0.0260. \citet{arieli2018robust} modify the scheme slightly\footnote{In the modified version, $\hat{\mu}=\begin{cases}
	0.49(f_1+f_2)& f_1+f_2\leq 1\\
	0.49(f_1+f_2)+0.02 & f_1+f_2> 1\\
\end{cases}$} to reduce the regret estimation to 0.0250.

\paragraph{Reduction of Regret Computation}

The analysis of $a_{ae}$ is more complicated because it involves a higher level information $p_1,p_2$. We will use \Cref{claim:fix} to reduce the regret computation to bilinear programming with 3 linear constraints for each side, such that we can reduce the support size of $Q$. 

\begin{proposition}\label{lem:basis}
	When $S_1$ and $S_2$ are independent conditioning on $W$, given any aggregator $a$, for all $Q$ over $\Sigma^2\times\{0,1\}$, there exists a prior $Q'$\footnote{The prior $Q'$ can depend on the aggregator $a$ and the prior $Q$.} over $\{A,B,C\}^2\times\{0,1\}$ such that $\E_Q((a(f_1,f_2,p_1,p_2)-f^*)^2)\leq \E_{Q'}((a(f_1,f_2,p_1,p_2)-f^*)^2)$.
\end{proposition}

\begin{corollary}
	In the conditionally independent setting, for all aggregator $a$, $L(a) = \sup_{Q\in \mathcal{Q}_{9}} L(a,Q)$ where $\mathcal{Q}_{9}$ is the set of all conditionally independent priors over $\{A,B,C\}^2\times\{0,1\}$. 
\end{corollary}

We use $\mathcal{Q}_{9}$ because each $Q\in \mathcal{Q}_{9}$ is controlled by 9 parameters $\mu, \Pr[S_i=\sigma|W=w], i=1,2,\sigma=A,B,w=0,1$. Note that $\Pr[S_i=C|W=w]=1-\Pr[S_i=A|W=w]-\Pr[S_i=B|W=w]$.

\paragraph{Proof sketch} $Q'$ is a ``basic'' prior compared to $Q$. We will decompose the prior into convex combination of multiple basic priors and show that the regret of $Q$ must be smaller than the regret of a basic prior according to a multi-linear property of the regret.

To guarantee the multi-linear property of the regret (\Cref{claim:multilinear}), the decomposition can only happen when we fix multiple parameters of $Q$ such that given two agents' forecasts, $f_1,f_2$, both the optimal forecast $f^*$ and their expectations $p_1,p_2$ are fixed. In other words, we can only decompose $Q$ in a restricted space. 

\Cref{claim:fix} shows that the restricted space should have fixed $W$'s expectations, $o_i,z_i,i=1,2$. We then show that the restrictions can be translated into linear constraints for prior vectors $\mathbf{q}_i,i=1,2$. In detail, the constraints of the restricted space are 3 linear equations (\Cref{claim:linear}). This is why we pick the support size $|\Sigma|$ of the simple prior to be 3. %The description of the linear constraints rely on the fact that this is a conditionally independent setting, where it's sufficient to describe the relationship between two agents' signals by describing the relationship between each agent's signal and $W$. 

Finally, we use the fact that any feasible solution of a linear system can be written as a convex combination of basic feasible solutions to decompose $Q$ into multiple basic priors.

\begin{proof}[Proof of \Cref{lem:basis}]

%Given $Q$, let $\mu$ denote $\Pr_Q[W=1]$. For all $x\in\Sigma$, let $f_1^x$ denote agent 1's forecast conditioning on she receives signal $x$, $\Pr_Q[W=1|S_1=x]$. Let $f_2^x$ denote agent 2's forecast conditioning on she receives signal $x$, $\Pr_Q[W=1|S_2=x]$. Let $q_1^x$ denote the prior probability that agent 1 receives x, i.e., $\Pr_Q[S_1=x]$ and $q_2^x$ denote the prior probability that agent 2 receives x, i.e., $\Pr_Q[S_2=x]$. For $i=1,2$, let $\mathbf{f}_i=(f_i^x)_{x\in\Sigma}$ denote a $|\Sigma|$-dimensional \emph{forecast vector}. Let $\mathbf{q}_i=(q_i^x)_{x\in\Sigma}$ denote a $|\Sigma|$-dimensional \emph{prior vector}. 

%From the perspective of a person who knows the outcome $W$, let $z_i$ denote his expectation for agent $i$'s forecast conditioning on $W=0$, $\sum_x f_i^x \Pr_Q[S_i=x|W=0]$ and $o_i$ denote his expectation for agent $i$'s forecast conditioning on $W=1$, $\sum_x f_i^x \Pr_Q[S_i=x|W=1]$ for $i=1,2$. 

%\begin{restatable}{claim}{claimfix}\label{claim:fix}$p_1=(1-f_1) z_2 + f_1 o_2$, $p_2= f_2 o_1 + (1-f_2) z_1$, $\mu o_i + (1-\mu) z_i=\mu, i=1,2$ and $f^*=\frac{(1-\mu) f_1 f_2}{(1-\mu)f_1 f_2 + \mu(1-f_1)(1-f_2)}$. \end{restatable}

We first illustrate the multi-linearity. 

\begin{restatable}{claim}{claimmultilinear}\label{claim:multilinear}
	Fixing $\mu,o_1,o_2,\mathbf{f}_1,\mathbf{f}_2$, the regret is a multilinear function of $\mathbf{q}_1,\mathbf{q}_2$. Formally, there exists a function $\psi$ such that the regret is $\mathbf{q}_1^{\top} \mathbf{\Psi} \mathbf{q}_2$ where $\Psi_{s_1,s_2}=\psi(f_1^{s_1},f_2^{s_2},\mu,o_1,o_2),\forall s_1,s_2$.
\end{restatable}

Proof of \Cref{claim:multilinear} is deferred to \Cref{sec:proofs}. Multi-linearity holds only in the restricted space. Let the restricted set of priors $\mathcal{Q}_{\mu,o_1,o_2,\mathbf{f}_1,\mathbf{f}_2}$ be the set of priors whose $\mu,o_1,o_2,\mathbf{f}_1,\mathbf{f}_2$ are the same as $Q$. Note that the prior vector $\tilde{\mathbf{q}}_i,i=1,2$ of $\tilde{Q}\in \mathcal{Q}_{\mu,o_1,o_2,\mathbf{f}_1,\mathbf{f}_2}$ can be different from the prior vector $\mathbf{q}_i,i=1,2$ of $Q$, though we cannot pick the prior vector $\tilde{\mathbf{q}}_i,i=1,2$ arbitrarily. \Cref{claim:linear} shows that the restrictions for the prior vectors can be translated into linear constraints.

%4) \sum_x  f_i^x \frac{(1-f_i^x) \tilde{q}_i^x}{1-\mu} = z_i

\begin{restatable}{claim}{claimlinear}\label{claim:linear}

For all $\tilde{Q}\in \mathcal{Q}_{\mu,o_1,o_2,\mathbf{f}_1,\mathbf{f}_2}$, for $i=1,2$, $\tilde{\mathbf{q}}_i$ satisfy the following linear constraints:
\[\begin{cases}1) \sum_x \tilde{q}_i^x = 1, & 2) \sum_x f_i^x \tilde{q}_i^x = \mu \\
3) \sum_x \tilde{q}_i^x \frac{(f_i^x)^2 }{\mu} = o_i \end{cases}\]
and the constraint that $\tilde{\mathbf{q}}_i$'s elements are non-negative. Let $C_{\mu,o_i,\mathbf{f}_i}$ denote the set of all possible solutions that satisfy the above constraints. Any pair of prior vectors that satisfies the above constraints, i.e., any $(\hat{\mathbf{q}}_1,\hat{\mathbf{q}}_2)\in C_{\mu,o_1,\mathbf{f}_1}\times C_{\mu,o_2,\mathbf{f}_2}$, corresponds to a joint distribution $\hat{Q}$ in $\mathcal{Q}_{\mu,o_1,o_2,\mathbf{f}_1,\mathbf{f}_2}$. 
\end{restatable}

%Because $\mu o_i + (1-\mu) z_i=\mu, i=1,2$, we have $\sum_x  f_i^x \frac{(1-f_i^x) \tilde{q}_i^x}{1-\mu}=\frac{1}{1-\mu} (\mu-\mu o_i)=z_i$. Thus, equation 3) is satisfied automatically with equation 2) and 3) and there are essentially 3 constraints for each $\tilde{\mathbf{q}}_i$.  

Proof of \Cref{claim:linear} is deferred to \Cref{sec:proofs}.

Finally, we decompose the prior $Q$ into basic priors. Any solution in $C_{\mu,o_i,\mathbf{f}_i}$ can be written as a convex combination of basic feasible solutions which have $\leq 3$ non-zero entries \citep{luenberger1973introduction}. We call the prior vector that corresponds to a basic feasible solution the basic prior vector. 

According to \Cref{claim:multilinear}, there exists two basic prior vectors $\mathbf{b}_1,\mathbf{b}_2$ where the regret of $Q$ will be $\leq$ $\mathbf{b}_1^{\top} \mathbf{\Psi} \mathbf{b}_2$. 

According to \Cref{claim:linear}, we can use $\mathbf{b}_1,\mathbf{b}_2$ to construct a valid prior $Q'\in \mathcal{Q}_{\mu,o_1,o_2,\mathbf{f}_1,\mathbf{f}_2}$ whose prior vectors are  $\mathbf{b}_1,\mathbf{b}_2$. Thus, $Q'$'s regret is $\mathbf{b}_1^{\top} \mathbf{\Psi} \mathbf{b}_2$. For each $i=1,2$, we define the support of $Q'$ for agent $i$'s private signals as the set of all signals $\sigma$ where $\Pr_{Q'}[S_i=\sigma]>0$. Because the prior vectors of $Q'$ are the basis vectors  $\mathbf{b}_1,\mathbf{b}_2$. The size of $Q'$'s support for agent $i=1,2$ is $\leq 3$. We can relabel the signals in the support as elements in $\{A,B,C\}$. 

Therefore, we have constructed a prior $Q'$ over $\{A,B,C\}^2\times\{0,1\}$ such that $\E_Q((a(f_1,f_2,p_1,p_2)-f^*)^2)\leq \E_{Q'}((a(f_1,f_2,p_1,p_2)-f^*)^2)$.

\end{proof}

\paragraph{Numerical Results}

Aided by the above reduction results, we can estimate the regret of Average Expectation $a_{ae}$ and compared to other schemes. This may not be a fair comparison for previous work because they propose the schemes in a different setting and some of them work well for at least moderate number of forecasters. Thus, in addition to the Simple Average and Average-prior, for reference, we illustrate the results for two additional natural schemes, Minimal Pivoting \citep{palley2019extracting}, Meta-prob Weighting\footnote{If both $|f_1-p_1|$ and $|f_2-p_2|$ are zero, the scheme reduces to simple average scheme.} \citep{martinie2020using}, that do not heavily rely on the number of forecasters and the type of the information structure. We use Matlab to estimate the maximum of their corresponding functions. We defer the estimating details to appendix.

\paragraph{Asymmetric vs. Symmetric} It is worth noting that the first two schemes, Simple Average and Average-prior, obtain the maximum in the most asymmetric case when one forecaster has no information. In contrast, the schemes that involve peer expectation, i.e., Average Expectation, Minimal Pivoting and Meta-prob Weighting, obtain the maximum in the symmetric case. In fact, unlike Simple Average and Average-prior, there is no regret in the most asymmetric case for Average Expectation, Minimal Pivoting and Meta-prob Weighting. In the most asymmetric case, the uninformative forecaster's forecast and expectation are prior. The informative forecaster's expectation for the uninformative forecaster's forecast is also the prior. That is, if agent 1 is uninformative, $f_1=p_1=\mu$, $p_2=\mu$. Average Expectation, Minimal Pivoting and Meta-prob Weighting will aggregate the forecasts to $f_2$, which is the optimal forecast as well. Thus, they have no regret in the most asymmetric case.

\begin{table*}
\caption{Regret: Two conditionally independent forecasters}\label{table:compare}
	\begin{center}
\begin{tabular}{ c|c|c } 
 \hline
Aggregator &  Formula &  Regret Estimation\\
 \hline 
 \hline
Simple Average &  $\frac{f_1+f_2}{2}$ &  0.0625* \\
 &   &   \\
 \hline 
Average-prior &  $\frac{(1-\frac{f_1+f_2}{2}) f_1 f_2}{(1-\frac{f_1+f_2}{2})f_1 f_2 + \frac{f_1+f_2}{2}(1-f_1)(1-f_2)}$ &   \\
\citep{arieli2018robust} &   &  0.0260* \\
 \hline 
Average Expectation &  $\frac{(1-\frac{p_1+p_2}{2}) f_1 f_2}{(1-\frac{p_1+p_2}{2})f_1 f_2 + \frac{p_1+p_2}{2}(1-f_1)(1-f_2)}$ &  0.0072* \\
 &   &   \\
 \hline 
Minimal Pivoting &  $f_1+f_2-\frac{p_1+p_2}{2}$ & 0.0394* \\
\citep{palley2019extracting} &   &   \\

 \hline 
 Meta-prob Weighting &  $\frac{|f_1-p_1|}{|f_2-p_2|+|f_1-p_1|}f_1+\frac{|f_2-p_2|}{|f_2-p_2|+|f_1-p_1|}f_2$ &  0.0556* \\
 \citep{martinie2020using} &   &   \\

 \hline
\end{tabular}
\end{center}

\end{table*}

%The intuition behind this estimation becomes clear when we discuss higher level of expectations in later section.  We will define the level zero expectations as agents' forecasts. The peer expectations are level one expectation, which is also the expectations of the level zero expectations. As the number of levels increase, we will show that they converge to the prior $\mu$. This provides insight that in estimating $\mu$, why using level one expectations, the peer expectations, can be better than level zero expectations, the forecasts. 

\subsection{Upperbound for C.I.I.D.} This section considers a natural setting where the two agents' private signals are i.i.d conditioning on $W$. The above mutli-linearity technique cannot be extended to the c.i.i.d. setting because multi-linearity does not hold when $Q$ must be symmetric for the two agents. Therefore, a non-trivial upperbound analysis for c.i.i.d. remains to be an open question in the original setting \citep{arieli2018robust}. 

%\citet{arieli2018robust} reduce the support size of $Q$ to 2 for the setting of two conditionally independent agents, by representing the regret as a format that is multilinear in $Q$'s marginal. 

Luckily, in the c.i.i.d. setting with peer expectations, we can construct aggregators whose regret is zero when agents' forecasts are different. This allows a tractable upper bound analysis. 

\paragraph{C.I.I.D. Aggregator}
\begin{restatable}{claim}{claimiid}\label{claim:iid}
	In the c.i.i.d. setting, when $f_1\neq f_2$, $\mu =\frac{f_2 p_1 - f_1 p_2}{f_2 - f_1 - p_2 + p_1}$.  \end{restatable}
	
Proof of \Cref{claim:iid} is deferred to \Cref{sec:proofs}.

\begin{definition}[C.I.I.D. Aggregator] \label{def:iid}
\[ a_{ciid}(f_1,f_2,p_1,p_2) = \frac{(1-\hat{\mu}) f_1 f_2}{(1-\hat{\mu})f_1 f_2 + \hat{\mu}(1-f_1)(1-f_2)}\]
where $\hat{\mu} = \begin{cases}\frac{f_2 p_1 - f_1 p_2}{f_2 - f_1 - p_2 + p_1} & f_1\neq f_2\\ \frac{p_1+p_2}{2}=p_1=p_2 & f_1=f_2 \end{cases}$
\end{definition}

\paragraph{Reduction of Regret Computation} We will reduce the support size of $Q$  in the regret computation of the above aggregator. The reduction works because that $a_{ciid}$ has no regret when two forecasters report distinct forecasts (\Cref{claim:iid}).

\begin{proposition}\label{lem:iid}
	For all $Q$ over $\Sigma^2\times\{0,1\}$, there exists a prior $Q'$ over $\{A,B,C\}^2\times\{0,1\}$ such that $\E_Q((a_{ciid}(f_1,f_2,p_1,p_2)-f^*)^2)\leq \E_{Q'}((a_{ciid}(f_1,f_2,p_1,p_2)-f^*)^2)$.
\end{proposition}

\begin{corollary}
	In the c.i.i.d. setting, $L(a_{ciid}) = \sup_{Q\in \mathcal{Q}_{5}} L(a_{ciid},Q)$ where $\mathcal{Q}_{5}$ is the set of all c.i.i.d. conditionally independent priors over $\{A,B,C\}^2\times\{0,1\}$. 
\end{corollary}

We use $\mathcal{Q}_{5}$ because each $Q\in \mathcal{Q}_{5}$ is controlled by 5 parameters $f^A,f^B,f^C$ and $q^A, q^B$. Note that $q^C=1-q^A-q^B$. We omit the subscript $i$ because agent's identity does not matter in the c.i.i.d. setting. The regret estimation over $\mathcal{Q}_5$ for $a_{ciid}$ is 0.00391*. 

\paragraph{Proof sketch} Because the regret only happens when two forecasters report the same forecast, the $\mathbf{\Psi}$ matrix used in \Cref{claim:multilinear} is a diagonal matrix with non-negative elements. In this case, maximizing over symmetric prior vectors is equivalent to maximizing over asymmetric prior vectors, which allows us to reduce the regret maximization to a bilinear programming again. 

\begin{proof}[Proof of \Cref{lem:iid}] In the c.i.i.d. setting, because agents' identities does not matter, we omit the subscript $i$ in notations defined in \Cref{sec:obs}. Without loss of generality, we assume that $f^s\neq f^{s'}$ for $s\neq s'$. Otherwise, we can merge the two signals because $f^*$ only depends on the forecasts and $\mu$ in the conditionally independent case.

\Cref{claim:iid} shows that when $f_1\neq f_2$, $\hat{\mu}=\mu$. Thus, when $f_1\neq f_2$, $(a_{ciid}(f_1,f_2,p_1,p_2)-f^*)^2=0$.

	\begin{align*}
		 & \E_Q (a_{ciid}(f_1,f_2,p_1,p_2)-f^*)^2 \\ \tag{zero regret for $f_1\neq f_2$ and we have assumed $f^s\neq f^{s'}$ for $s\neq s'$ w.l.o.g.}
		= & \sum_s \Pr_Q[S_1=S_2=s] (a_{ciid}(f_1,f_2,p_1,p_2)-f^*)^2	\\
	\end{align*}
	
	Note that 
	
	\begin{align*}
		&\Pr_Q[S_1=S_2=s]\\
		=&\Pr_Q[S_1=S_2=s|W=1]\Pr_Q[W=1]+\Pr_Q[S_1=S_2=s|W=0]\Pr_Q[W=0]\\
		= & (\Pr_Q[S_1=s|W=1])^2\mu+(\Pr_Q[S_1=s|W=0])^2(1-\mu)\\
		= & (\frac{q^s f^s}{\mu})^2\mu+(\frac{q^s (1-f^s)}{1-\mu})^2(1-\mu)
	\end{align*}

Thus, 	
	\begin{align*}
	& \E_Q (a_{ciid}(f_1,f_2,p_1,p_2)-f^*)^2 \\
		= & \sum_s (q^s)^2 (\frac{(f^s)^2}{\mu}+\frac{(1-f^s)^2}{1-\mu})\\ \tag{$\hat{\mu}=p_1=p_2=f^s o+ (1-f^s)z$ (\Cref{claim:fix})}
		&\times (\frac{(1-(f^s o + (1-f^s) z))(f^s)^2}{(1-(f^s o + (1-f^s) z))(f^s)^2+(f^s o + (1-f^s) z) (1-f^s)^2}\\
		&\ \ \ \ \ \ -\frac{(1-\mu)(f^s)^2}{(1-\mu)(f^s)^2+\mu (1-f^s)^2})\\
		= & \sum_s (q^s)^2 \phi(f^s,\mu,o,z)\\ \tag{$C_{\mu,o,\mathbf{f}}$	denotes the set of $\mathbf{\tilde{q}}$ where 
	$\begin{cases}1) \sum_x \tilde{q}^x = 1, & 2) \sum_x f^x \tilde{q}^x = \mu \\
3) \sum_x f^x \frac{f^x \tilde{q}^x }{\mu} = o
\end{cases}$}
		\leq & \max_{\mathbf{q}_1,\mathbf{q}_2\in C_{\mu,o,\mathbf{f}}} \sum_s q_1^s q_2^s \phi(f^s,\mu,o,z)\\ \tag{$\mathbf{b},\mathbf{b}'$ are basis vectors in $C_{\mu,o,\mathbf{f}}$}
		\leq & \sum_s b^s b'^s \phi(f^s,\mu,o,z)\\ \tag{$\mathbf{v}(s)=b^s \sqrt{\phi(f^s,\mu,o,z)},\mathbf{v'}(s)=b'^s \sqrt{\phi(f^s,\mu,o,z)},\forall s$ }
		= & \mathbf{v} \cdot \mathbf{v}'\\
		\leq & \max\{ ||\mathbf{v}||^2, ||\mathbf{v}'||^2 \}\\
		= &\max\{\sum_s ( b^s)^2  \phi(f^s,\mu,o,z),\sum_s ( b'^s)^2  \phi(f^s,\mu,o,z)\}
		\end{align*}
		
Without loss of generality, we assume $\sum_s ( b'^s)^2  \phi(f^s,\mu,o,z)\geq \sum_s ( b^s)^2  \phi(f^s,\mu,o,z)$. According to \Cref{claim:linear}, we can pick both prior vectors as $\mathbf{b}'$ and use $\mu,\mathbf{f}_i=\mathbf{f},o_i=o,i=1,2$ to construct symmetric $Q'$ over $\{A,B,C\}^2\times \{0,1\}$ and \[\E_Q((a_{ciid}(f_1,f_2,p_1,p_2)-f^*)^2)\leq \sum_s ( b_1^s)^2  \phi(f^s,\mu,o,z)=\E_{Q'}((a_{ciid}(f_1,f_2,p_1,p_2)-f^*)^2).\]

\end{proof}

\subsection{Lowerbound}

With the above upperbound results, it's left to analyze the limitation. We follow \citet{arieli2018robust} and model the setting as a zero-sum game between nature and aggregator. We construct two symmetric priors $Q,Q'$ such that nature will play a mixed strategy to choose the prior. Any aggregator will have a non-trivial regret even if the mixed strategy is given. This regret will provide a lower bound for all possible aggregators. 

We want $Q,Q'$ to be simple such that the analysis is tractable. On the other hand, $Q,Q'$ should have sufficient delicacy. This allows us to design them carefully such that the aggregator cannot infer the prior nature actually chose in some scenarios, even if the aggregator has the information of the forecasts and the peer expectations. This will be more complicated than the original setting where the aggregator only has the forecasts. 

\paragraph{Zero-sum Game} In detail, nature will pick $Q$ with probability $x$ and $Q'$ with probability $1-x$. Let $x Q + (1-x) Q'$ denote the mixed prior. Nature's choice is a private message shared by agents. That is, agents know nature's choice while the aggregator does not. 

In this setting, without knowing nature's choice, the optimal aggregator is the Bayesian aggregator $\Pr_{x Q + (1-x) Q'}[W=1|f_1,f_2,p_1,p_2]$ who predicts $W$ based on $f_1,f_2,p_1,p_2$ and the fact that nature will pick $Q$ with probability $x$ and $Q'$ with probability $1-x$. 

\begin{observation}\label{obs:mix}
	For all $a$, the regret of $a$ is at least the regret of the Bayesian aggregator under the mixed prior $x Q + (1-x) Q'$ setting:
	\begin{align*}
		L(a)&\geq \E_{x Q + (1-x) Q'} (\Pr_{x Q + (1-x) Q'}[W=1|f_1,f_2,p_1,p_2]-f^*)^2\\
		&=x \E_{Q} \left(\Pr_{x Q + (1-x) Q'}[W=1|f_1,f_2,p_1,p_2]-\Pr_{Q}[W=1|f_1,f_2,p_1,p_2]\right )^2 \\
		 & + (1-x) \E_{Q'} \left(\Pr_{x Q + (1-x) Q'}[W=1|f_1,f_2,p_1,p_2]-\Pr_{Q'}[W=1|f_1,f_2,p_1,p_2]\right)^2
	\end{align*}
\end{observation}

%L(a)=\sup_Q L(a,Q) = \sup_Q \E_Q(a(f_1,f_2,p_1,p_2)-f^*)^2

\begin{proof}[Proof of \Cref{obs:mix}]
	\begin{align*}
		&x \E_{Q} \left(\Pr_{x Q + (1-x) Q'}[W=1|f_1,f_2,p_1,p_2]-\Pr_{Q}[W=1|f_1,f_2,p_1,p_2]\right )^2 \\
		 & + (1-x) \E_{Q'} \left(\Pr_{x Q + (1-x) Q'}[W=1|f_1,f_2,p_1,p_2]-\Pr_{Q'}[W=1|f_1,f_2,p_1,p_2]\right)^2\\
		\leq &x \E_{Q} \left(a(f_1,f_2,p_1,p_2)-\Pr_{Q}[W=1|f_1,f_2,p_1,p_2]\right)^2 \\
		 & + (1-x) \E_{Q'} \left(a(f_1,f_2,p_1,p_2)-\Pr_{Q'}[W=1|f_1,f_2,p_1,p_2]\right)^2\\		\leq & \sup_{P} \E_{P}\left(a(f_1,f_2,p_1,p_2)-\Pr_{P}[W=1|f_1,f_2,p_1,p_2] \right)^2 = L(a)
	\end{align*}
\end{proof}

Both $Q,Q'$ are symmetric and have the same support $\{A,B,C\}$ for each agent's private signals.

\begin{restatable}{lemma}{lemmamix}\label{lemma:mix}
	There exists proper $Q,Q'$ such that \[\E_{x Q + (1-x) Q'} (\Pr_{x Q + (1-x) Q'}[W=1|f_1,f_2,p_1,p_2]-f^*)^2> 0.00144\]  \end{restatable}

With \Cref{obs:mix}, we directly obtain the following lower bound result. 

	\begin{corollary}
		In the conditionally independent setting and the c.i.i.d. setting, for all aggregator $a$, $L(a)\geq 0.00144$. 
	\end{corollary}
 
\paragraph{Proof sketch}  We first illustrate our choice of priors which are simple yet sufficiently delicate for tuning. 

\paragraph{Prior Design} Let $\mu=\Pr_{Q}[W=1]$ and $\mu'=\Pr_{Q'}[W=1]$. We design the priors such that $f^A_1=f^A_2=0$, $f^B_1=f^B_2=b$, $f^C_1=f^C_2=1$ for both priors. Conditioning on agent 1 receives signal $B$ and the prior is $Q$, agent 1's expectation for agent 2's forecast is $p^B_1|Q$. $Q$ and $Q'$ will be designed carefully such that $p^B_1|Q=p^B_1|Q'=p$. Because the priors are symmetric, $p^B_2|Q=p^B_2|Q'=p$. The above requirements will be reduced to restrictions for $q^A,q^B,q^C, q'^{A},q'^{B},q'^{C}$.
 
 %By calculations, with the above prior, $\E_{x Q + (1-x) Q'} (\Pr_{x Q + (1-x) Q'}[W=1|f_1,f_2,p_1,p_2]-f^*)^2> 0.00144$.

Because of the above design, when both agents receive B, they will report $(f_1=b,f_2=b,p_1=p,p_2=p)$ regardless of the prior. In this case, no aggregator, including the Bayesian aggregator, can identify the prior nature chose. Thus, we use the regret of the Bayesian aggregator in the case when both agents receive $B$ as a lower bound.

\paragraph{Optimal Parameters} We illustrate how we find the optimal $\mu,\mu',b,p,x$ to maximize the regret of the Bayesian aggregator in the above case.  
\begin{description}
	\item [Step 1] We show that the lower bound is a function of $\mu,\mu',b,p,x$ which can be written as the following format: $(p-b)^2\psi(\mu,\mu',b,x)$;
	\item [Step 2] by assuming $b< \mu<\mu'$, we show that $Q,Q'$ can be valid only if $p-b\leq b(1-b\frac{1}{\mu})$. It's left to maximize $(b(1-b\frac{1}{\mu}))^2 \psi (\mu,\mu',b,x)$;
	\item [Step 3] we then fix $\mu,\mu',b$ and substitute $x$ by the solution of $\frac{\partial\psi}{\partial x}(x)=0 $;
	\item [Step 4] finally, we optimize over $\mu,\mu',b$.
\end{description}

After the above steps, we pick $\mu = 0.7427, \mu' = 0.9100, b=0.62, p=0.7224, x = 0.1991$ and construct the following $Q,Q'$ that satisfies the above conditions:
\begin{center}
\begin{tabular}{ |c| c c |c| c c |}
\hline
 $Q\approx$  & 1 & 0 & $Q'\approx$ & 1 & 0\\ 
  \hline 
 A  & 0 & 0 & A & 0 &  0.0443\\  
 B & 0.6770 $\times$ 0.62 & 0.6770 $\times$ 0.38 & B & 0.1228 $\times$ 0.62 & 0.1228 $\times$ 0.38  \\
 C  & 0.3230 & 0 & C &  0.8339 & 0\\
 \hline 
\end{tabular}
\end{center}

Finally, the above parameters lead to regret $>0.00144$, which provides a valid lower bound according to \Cref{obs:mix}. Formal proof is deferred to \Cref{sec:proofs}.

%\geq (p-b)^2\psi(\mu,\mu',b,x)

%Note that $\mu=\sum q^s f^s$. Thus, $q^s f^s\leq \mu\leq q^s f^s+1-q^s$. 

%\leq & (q^s)^n (\frac{(f^s)^n}{(q^s f^s)^{n-1}}+\frac{(1-f^s)^n}{(q^s(1-f^s))^{n-1}})\\

%\begin{align*}	\Pr[\forall i, S_i=s] =& (q^s)^n (\frac{(f^s)^n}{\mu^{n-1}}+\frac{(1-f^s)^n}{(1-\mu)^{n-1}})\\	\end{align*}

\subsection{Better Upperbound with Weights \& Hard Sigmoid}

In the original setting, using average forecasts as a proxy for prior has already performed relatively well (0.0260) compared to the lower bound (0.0225). However, unlike the original setting, there exists a certain amount of gap between the performance (0.0072) of $a_{ae}$, which uses the average expectation as a proxy for prior, and the lower bound (0.00144). Therefore, we need a more sophisticated proxy for prior in the new setting. 

We first try to use the trick in \citet{arieli2018robust} to improve the upperbound. They change $\hat{\mu}=\frac{f_1+f_2}{2}$ to $\hat{\mu}=\begin{cases}
	0.49(f_1+f_2)& f_1+f_2\leq 1\\
	0.49(f_1+f_2)+0.02 & f_1+f_2> 1\\
\end{cases}$ and reduce the regret estimation from $0.0260$ to $0.0250$. We replace $f_i$ by $p_i$. However, this improves slightly (e.g. in c.i.i.d. setting: 0.0039 $\rightarrow$ 0.0037). 

Both the forecast based proxy and the peer expectation based proxy ignore the relationship between the forecast and the expectation. Thus, we will introduce a new aggregator that utilizes the relationship. The following observation inspires the new aggregator. 

\begin{observation}\label{obs:hs}
	When $|o_i-z_i|< 1$ for $i=1,2$, $f_i=p_i\Rightarrow f_i=p_i =\mu$ for $i=1,2$. Moreover, 
	\[p_1-\mu = (f_1-\mu)(o_2-z_2)\Rightarrow \mu = \frac{p_1-(o_2-z_2) f_1}{1-(o_2-z_2)}\] and we have analogous result for $p_2,f_2$. When $|o_i-z_i|= 1$, agent $i$ is a perfect forecaster who knows $W$. 
\end{observation}

Note that $|o_i-z_i|\leq 1$ for $i=1,2$. The above observation implies that $p_i$ is closer to $\mu$ than $f_i$, for $i=1,2$. This explains that why $\frac{p_1+p_2}{2}$ is a better proxy for $\mu$ than the proxy $\frac{f_1+f_2}{2}$ in the original setting. However, these two proxies ignore the relationship between $f_1,f_2$ and $p_1,p_2$. We utilize the fact $f_i=p_i\Rightarrow f_i=p_i =\mu$ by replacing the simple average expectation proxy to the weighted average expectation proxy. Note that we do not need to worry about the $|o_i-z_i|= 1$ case where $f_i=p_i\Rightarrow f_i=p_i =\mu$ does not hold. In this case, there is no regret because agent $i$ provides a perfect forecast 0 or 1. 

\begin{proof}[Proof of \Cref{obs:hs}]
	According to \Cref{claim:fix}, $p_1=f_1 o_2 + (1-f_1) z_2$ and $\mu = \mu o_2 + (1-\mu) z_2 $. Thus 

\[p_1-\mu = (f_1-\mu)(o_2-z_2).\] When $o_2-z_2\neq 1$, $f_1=p_1\Rightarrow f_1=p_1 =\mu$. 

By one more manipulation of the above equation, 
\[\mu = \frac{p_1-(o_2-z_2) f_1}{1-(o_2-z_2)}. \]

Additionally, when $o_i-z_i=1$, $o_i=1$ and $z_i=0$ because $o_i,z_i\in [0,1]$. In this case, when $W=1$, agent $i$ will report $f_i=1$, and when $W=0$, agent $i$ will report $f_i=0$. In other words, agent $i$ is a perfect forecaster. 

%Moreover, because we have assumed $\mu\in(0,1)$, $o_i-z_i<1$. Otherwise, $o_i-z_i=1\Rightarrow o_i = 1, z_i=0$. $o_i=\sum_s \Pr[S_i=s|W=1] f_i^s=\sum_s \frac{q_i^s f_i^s}{\mu} f_i^s$. This implies that $\forall s$, $f_i^s=0$ or $1$. 
\end{proof}

%\yk{add reference to the definitions in the introduction}

\paragraph{Better Upperbound with Weighting}

Let's consider the following example. 

\begin{example}
	We receive $f_1=p_1=0.6$ and $f_2=0.3, p_2=0.4$. In this case, $\frac{f_1+f_2}{2}=0.45, \frac{p_1+p_2}{2}=0.5$, while based on \Cref{obs:hs}, $\mu = f_1=p_1=0.6$. 
\end{example}

The above example naturally implies the following aggregator:

\begin{definition}[Weighted Expectation Aggregator] \label{def:we}
\[ a_{we}(f_1,f_2,p_1,p_2) = \frac{(1-\hat{\mu}) f_1 f_2}{(1-\hat{\mu})f_1 f_2 + \hat{\mu}(1-f_1)(1-f_2)}\]
where $\hat{\mu} = \frac{|f_2-p_2|^\alpha}{|f_1-p_1|^\alpha+|f_2-p_2|^\alpha} p_1 + \frac{|f_1-p_1|^\alpha}{|f_1-p_1|^\alpha+|f_2-p_2|^\alpha} p_2$. If both $|f_1-p_1|$ and $|f_2-p_2|$ are zero, $\hat{\mu}=\frac{p_1+p_2}{2}$.\end{definition}

The format is similar to Meta-prob Weighting (\Cref{table:compare}). Nonetheless, note that Meta-prob Weighting assigns more weight to the forecaster with higher $|f_i-p_i|$, while we do the opposite because we are estimating the prior. We have tried different $\alpha$ numerically and $\alpha=0.5$ is the best. By picking $\alpha=0.5$, the regret estimation of $a_{wp}$ is 0.0040.

\paragraph{Better Upperbound with Hard Sigmoid (C.I.I.D.)} The above weighted average trick does not work for the c.i.i.d. case. In this c.i.i.d. case, the regret only happens when $f_1=f_2$, which also leads to the same expectations $p_1=p_2$ (\Cref{claim:fix}). In this case, either simple average or weighted average leads to the same value, $p_1=p_2$. Here we introduce another trick, the hard sigmoid function. 

At a high level, \Cref{obs:hs} implies that if we receive $f_1=0.3, p_1=0.5$, $\mu$ will be strictly greater than $p_1$. The higher the $|f_1-p_1|$ is, the further the $\mu$ is from $p_1$. Therefore, directly using $p_1$ to estimate $\mu$ may not be a good choice. 

In detail, in c.i.i.d., $o=o_i,z=z_i,i=1,2$. According to \Cref{obs:hs}, if we know $o-z$, we can perfectly obtain $\mu=\frac{p_1-(o-z) f_1}{1-(o-z)}$. Without $o-z$, fixing $f_1$, we also know that $\mu(p_1)$ is a linear function with slope $\geq 1$ and $f_1$ as the fixed point.  

To this end, we design another proxy for $\mu$ which involves both $p_1$ and $f_1$: $\frac{p_1-c_{oz}f_1}{1-c_{oz}}$ where $c_{oz}$, the proxy for $o-z$, is a constant we can pick later. To avoid the new proxy to obtain value outside $[0,1]$, we also bound it within $[c_{lb},c_{ub}]$. Finally, we have a family of proxy functions what are hard sigmoid functions:

\[ \max\{\min\{\frac{p_1-c_{oz}f_1}{1-c_{oz}},c_{ub}\},c_{lb}\}\]

We then pick proper constants to reduce the regret. 

\begin{definition}[Hard Sigmoid Aggregator] \label{def:hs}
\[ a_{hs}(f_1,f_2,p_1,p_2) = \frac{(1-\hat{\mu}) f_1 f_2}{(1-\hat{\mu})f_1 f_2 + \hat{\mu}(1-f_1)(1-f_2)}\]
where $\hat{\mu} = \begin{cases}\frac{f_2 p_1 - f_1 p_2}{f_2 - f_1 - p_2 + p_1} & f_1\neq f_2\\ \max\{\min\{\frac{p_1-0.265 f_1}{1-0.265},0.96\},0.04\} & f_1=f_2 \end{cases}$
\end{definition}

Note that $a_{hs}$ has no regret when $f_1\neq f_2$. Thus, \Cref{lem:iid} still holds for $a_{hs}$. The regret estimation over $\mathcal{Q}_5$ for $a_{hs}$ is 0.00211*.

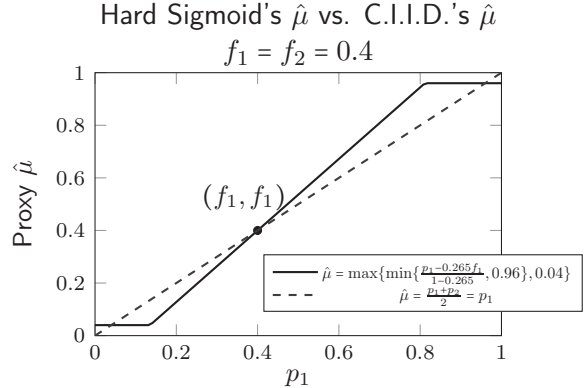
\begin{center}
	\begin{tikzpicture}
\sffamily
\begin{axis}[ 
    title = {Hard Sigmoid's $\hat{\mu}$ vs. C.I.I.D.'s $\hat{\mu}$\\ $f_1=f_2=0.4$}, 
    title style={align=center,yshift=-.1in}, 
    legend style={font=\small, 
        nodes={scale=0.6, transform shape}, 
        at={(1.2,.2)},
        anchor=east},
    width = 2.75in, height = 2.0in, 
    ylabel near ticks,
    ylabel = {\small Proxy $\hat{\mu}$},
    xlabel near ticks,
    ymax=1, ymin=0,% ymax=1, ymin=0.4,
    every tick label/.append style={font=\scriptsize},
    xmin=0,xmax=1,
    xlabel={\small $p_1$},
    xlabel style = {yshift=0.05in},
    ]
%\filldraw (1,1) circle[radius=1.5pt];
\node[fill, draw, circle, minimum width=3pt, inner sep=0pt] at (40,40) {}; 
\node[above right=10pt of {(18,35)}, outer sep=2pt,fill=white] {\small $(f_1,f_1)$};

\addplot [
    domain = 0:1,
    thick,
    samples = 100,
    color=black,
    ]
    {min(max((x-0.265*0.4)/(1-0.265),0.04),0.96)};
\addlegendentry{$\hat{\mu}=\max\{\min\{\frac{p_1-0.265 f_1}{1-0.265},0.96\},0.04\}$}

\addplot [
    domain = 0:1,
    thick,
    dashed,
    samples = 100,
    color=black!90,
    ]
    {x};
\addlegendentry{$\hat{\mu}=\frac{p_1+p_2}{2}=p_1$}
\end{axis}    
\end{tikzpicture} 

\end{center}

\paragraph{Better Upperbound with Weights + Hard Sigmoid} Finally, we combine the above two tricks in the conditionally independent setting. 

\begin{definition}[Weighted Hard Sigmoid Aggregator] \label{def:whs}

\[ a_{whs}(f_1,f_2,p_1,p_2) = \frac{(1-\hat{\mu}) f_1 f_2}{(1-\hat{\mu})f_1 f_2 + \hat{\mu}(1-f_1)(1-f_2)}\]
where $\hat{\mu} = \frac{|f_2-p_2|^\alpha}{|f_1-p_1|^\alpha+|f_2-p_2|^\alpha} \hat{\mu}_1 + \frac{|f_1-p_1|^\alpha}{|f_1-p_1|^\alpha+|f_2-p_2|^\alpha} \hat{\mu}_2$ and $\hat{\mu}_i = \max\{\min\{\frac{p_i-0.265 f_i}{1-0.265},0.96\},0.04\} $, $i=1,2$, $\alpha=0.5$. 	
\end{definition}

The regret estimation for the above aggregator is 0.00255*. Note that in the c.i.i.d. case, when $f_1=f_2$, the weighted hard sigmoid aggregator is the hard sigmoid aggregator.

\begin{table*}
\caption{Regret: Different Prior Proxy}\label{table:compare}
	\begin{center}
\begin{tabular}{ |c|c|c| } 
\hline
Aggregator &  Formula of Prior Proxy &  Regret Estimation\\
 \hline 
 \hline
Average-prior &  $\frac{f_1+f_2}{2}$ &   \\
\citep{arieli2018robust} &   &  0.0260* \\
 \hline 
Average Expectation &  $\frac{p_1+p_2}{2}$ &  0.0072* \\
 &   &   \\
 
 \hline 
Weighted Expectation&  $\frac{|f_2-p_2|^\alpha}{|f_1-p_1|^\alpha+|f_2-p_2|^\alpha} p_1 + \frac{|f_1-p_1|^\alpha}{|f_1-p_1|^\alpha+|f_2-p_2|^\alpha} p_2$ &  0.0040* \\
 &  $\alpha = 0.5$ &   \\
 \hline 
Weighted Hard Sigmoid  &  $\frac{|f_2-p_2|^\alpha}{|f_1-p_1|^\alpha+|f_2-p_2|^\alpha} \hat{\mu}_1 + \frac{|f_1-p_1|^\alpha}{|f_1-p_1|^\alpha+|f_2-p_2|^\alpha} \hat{\mu}_2$ &  0.0025* \\
 & $\hat{\mu}_i = \max\{\min\{\frac{p_i-0.265 f_i}{1-0.265},0.96\},0.04\},\alpha=0.5 $  &   \\

\hline
\hline 
C.I.I.D. &  $\frac{p_1+p_2}{2}, f_1=f_2$ &  0.0039* \\
 &  $\frac{f_2 p_1 - f_1 p_2}{f_2 - f_1 - p_2 + p_1},f_1\neq f_2$ &   \\

\hline 
Hard Sigmoid (C.I.I.D.) &  $\max\{\min\{\frac{p_i-0.265 f_i}{1-0.265},0.96\},0.04\}, f_1=f_2$ &  0.0021* \\
 & $\frac{f_2 p_1 - f_1 p_2}{f_2 - f_1 - p_2 + p_1},f_1\neq f_2$  &   \\

\hline
\end{tabular}
\end{center}

\end{table*}

\subsection{Extensions}

\paragraph{Conditionally Independent Conditioning on a Shared Signal} A slight generalization of the conditionally independent setting is the setting where the forecasters observe a common signal $S$ and their private signals $S_1,S_2$ are independent conditioning on $S$ and outcome $W$. This setting will have the same lower bound because it is more general. Regarding upper bound, it's easy to see that for all prior $Q$ in this setting, 
\[ L(a,Q)=\sum_s \Pr_{Q}[S=s] L(a,Q_s)\leq \max_s L(a,Q_s)\] where $Q_s$ denotes a joint distribution over $S_1,S_2,W$ conditioning on $S=s$. Thus, $Q_s$ is a prior in the conditionally independent setting. Therefore, this generalized setting's upper bound is the conditionally independent setting's upper bound as well.

\paragraph{Many C.I.I.D. Conditionally Independent Agents}

There are $n$ agents whose private signals are c.i.i.d. conditioning on $W$. With only agents' forecasts, \citet{arieli2018robust} implicitly show that for all aggregator $a$, there exists $Q$, such that even when the number of agents goes to infinite, $L(a,Q)\geq \frac{3}{32}$ (see \Cref{sec:manyregret}). In this setting, we can replace the peer expectation question by ``what's your expectation for the average of other forecasters' forecast?''. With peer expectations, for any fixed unknown $Q$, we can construct an aggregator whose regret on $Q$ goes to zero as the number of agents goes to infinite. 

\citet{chen2021wisdom} also show this result with a slightly different analysis. They assume infinite forecasters initially and do not need to consider the case when all forecasters report the same forecast. We consider the same forecast case and write down the regret of finite n. We show that the regret converges to zero. We still include the result for completeness. We find that even when all forecasters report the same forecast, the aggregated forecast using $a^n_{ciid}$ will converge to the optimal one. Formal analysis is deferred to \Cref{sec:proofs}.

\begin{definition}[c.i.i.d. aggregator for $n$ agents] 

\[ a^n_{ciid}(\{f_i,p_i\}) = \frac{(1-\hat{\mu})^{n-1}\Pi_i f_i}{(1-\hat{\mu})^{n-1}\Pi_i f_i + \hat{\mu}^{n-1}\Pi_i (1-f_i)}\]
where $\hat{\mu} = \begin{cases}\frac{f_j p_i - f_i p_j}{f_j - f_i - p_j + p_i} & \exists f_i\neq f_j\\ p_1 & \text{otherwise} \end{cases}$. \end{definition} 

\begin{restatable}{proposition}{propmany}\label{prop:many}
	For all $Q$, $\lim_{n\rightarrow \infty} L(a^n_{ciid},Q) = 0$.   \end{restatable}

\section{Higher Level Expectations}\label{sec:higher}

We have asked each agent ``what's your expectation for the other agent's forecast?''. Though it may not be practical, it's still interesting to ask what if we ask each agent ``what's your expectation for the other agent's expectation for your forecast?''. Formally, we define peer expectations of higher levels.  

\paragraph{Level $\ell$ Expectations \cite{samet1998iterated}} We define the level 1 expectation as agent's forecast, $p_i^1=f_i,i=1,2$. For all $\ell\geq 2$, we define $p_1^{\ell}$ as agent $1$'s expectation for $p_2^{\ell-1}$. $p_2^{\ell}$ is defined analogously. Note that the peer expectation $p_i = p_i^2, i=1,2$ is a level 2 expectation.

In the conditionally independent setting, with mild conditions, as $\ell$ goes to infinite, the expectation goes to the prior $\mu$. Thus, we can use the level $\ell$ expectation to construct an aggregator whose regret converges to zero. The idea is to use the average of the two agents' level $\ell$ expectation as the proxy for prior $\mu$. In general, as level $\ell$ goes to infinite, the expectation may not converge to $\mu$. For example, in the refinement-ordered setting, the expectation will converge to the less informed agent's forecast. But with proper non-degenerate assumptions, the level $\ell$ expectation also goes to the prior $\mu$ in general. Due to space limitations, details of the above results are deferred to Appendix~\ref{appendix:higher}.

\paragraph{Limitation} Even with an infinite level of expectations, we may not obtain optimal forecasts in some scenarios. When two agents' signals are independent bits where each bit is 0 with probability 0.5 and $W$ is their xor, the forecasts and all levels of expectations are 1/2. No scheme can aggregate optimally based on the forecasts and expectations. %Moreover, the level of expectations only depend on $\mathbf{f}_i,\mathbf{U},\mathbf{V}$. Thus, it only depends on the marginal joint distribution between each forecaster's private signal and outcome $W$, and the joint distribution between the forecasters' private signals. However, the optimal aggregation of agents' private signals require the knowledge of the joint distribution over $S_1,S_2$ and $W$. Thus, the forecasts and all levels of expectations do not provide sufficient information for perfect aggregation. 

\section{Conclusion and Discussion}

We analyze a new setting for one-shot forecast aggregation where each forecaster is additionally asked her expectation for the other forecaster's forecast. We analyze the setting in a framework of robust forecast aggregation and illustrate the possibility/impossibility of the additional information. We also show that when we can ask for higher level expectations, we can construct an aggregator which converges to the optimal one in the conditionally independent setting. 

Real-world experiments are natural future directions. Filling the gap is another natural question. Other interesting directions include:

\paragraph{Automated design} We can develop automated method to design the aggregator. For example, we can pick a proper family of aggregators and learn the parameters by playing the zero-sum game with nature. 
\paragraph{General Structure} In the worst case of general structure, the expectation will not help. For example, when two agents' signals are independent bits and $W$ is their xor, the forecasts and expectations are all 1/2. However, it's still interesting to analyze the benefit of expectation in general by replacing the definition of regret by the regret to the best forecaster.

\paragraph{Other Settings} It will be interesting to analyze the benefit of high order belief with an adversarial framework in other natural settings including decision making \citep{de2021robust,LEVY2022105075}, continuous outcome \citep{neyman2022you}. 

We hope this work can bring the adversarial framework to one-shot information aggregation with high order belief.

\newpage

\bibliographystyle{ACM-Reference-Format}
\bibliography{sample-bibliography}

% Appendix
\appendix

\section{Higher Level Expectations}\label{appendix:higher}

We have asked each agent ``what's your expectation for the other agent's forecast?''. Though it may not be practical, it's still interesting to ask what if we ask each agent ``what's your expectation for the other agent's expectation for your forecast?''.

Formally, we define peer expectations of higher levels. 

\paragraph{Level $\ell$ Expectations \cite{samet1998iterated}} We define the level 1 expectation as agent's forecast, $p_i^1=f_i,i=1,2$. For all $\ell\geq 2$, we define $p_1^{\ell}$ as agent $1$'s expectation for $p_2^{\ell-1}$. $p_2^{\ell}$ is defined analogously. Note that the peer expectation $p_i = p_i^2, i=1,2$ is a level 2 expectation.

In the conditionally independent setting, with mild conditions, as $\ell$ goes to infinite, the expectation goes to the prior $\mu$. Thus, we can use the level $\ell$ expectation to construct an aggregator whose regret converges to zero. 

\begin{definition}[Average Expectation Aggregator$^\ell$] 
\[ a_{ae}^{\ell}(f_1,f_2,p_1,p_2) = \frac{(1-\hat{\mu}) f_1 f_2}{(1-\hat{\mu})f_1 f_2 + \hat{\mu}(1-f_1)(1-f_2)}\]
where $\hat{\mu} = \frac{p_1^{\ell}+p_2^{\ell}}{2}$.%\begin{cases}\frac{f_2 p_1 - f_1 p_2}{f_2 - f_1 - p_2 + p_1} & f_1\neq f_2\\ \frac{p_1^k+p_2^k}{2} & f_1=f_2 \end{cases}$. 
\end{definition}

We focus on the setting where neither agent 1 or agent 2 is a perfect forecaster who knows $W$. Because if one of the forecaster is perfect, she will report 0 or 1 as forecast and we can easily achieve no regret by following her opinion. Moreover, note that if both agents are perfect, their peer expectations will always be $W$ rather than converge to $\mu$. 

\begin{restatable}{proposition}{prophigh}\label{prop:high}
	In the conditionally independent setting, where neither agent 1 or agent 2 is perfect, for $i=1,2$, $\lim_{\ell\rightarrow\infty} p_i^{\ell} = \mu $, thus \[\lim_{\ell\rightarrow\infty} L(a_{ciid}^{\ell},Q)=0.\] \end{restatable}

To show the above result, we first observe that $p_1^1=f_1=f_1\times 1 + (1-f_1)\times 0$, $p_1^2=f_1 o_2 + (1-f_1) z_2$. Thus, $p_1^2$ can be seen as a scaled $p_1^1$ within a smaller interval. Recursively, the interval becomes smaller and smaller and finally converges to $\mu$. Formal proof is deferred to \Cref{sec:proofs}.

In general, as level $\ell$ goes to infinite, the expectation may not converge to $\mu$. For example, in the refinement-ordered setting, the expectation will converge to the less informed agent's forecast. The following proposition describes the expectation series in general. 

Let a $|\Sigma|\times |\Sigma|$ matrix $\mathbf{U}$ record agent 1's prediction for agent 2's private signal. Formally, $\forall s,s'$, $U_{s,s'}=\Pr_Q[S_2=s'|S_1=s]$. Analogously, we define $V_{s,s'}=\Pr_Q[S_1=s'|S_2=s]$. Note that both $\mathbf{U}$ and $\mathbf{V}$ are row-stochastic matrices. 

Recall that $\mathbf{f}_i$ is a $|\Sigma|$-dimensional vector where $\mathbf{f}_i(s)=f^s_i=\Pr_Q[W=1|S_i=s]$. Let $\mathbf{p}_i^{\ell}$ be a $|\Sigma|$-dimensional vector where $\mathbf{p}_i^{\ell}(s)$ is agent i's level $\ell$ expectation conditioning on she receives private signal $s$.

\begin{restatable}{proposition}{propgeneral}\label{prop:general}
	In general, for integer $k\geq 0$, \[ \mathbf{p}_1^{2k+1} = (\mathbf{U}\mathbf{V})^k \mathbf{f}_1 \qquad \mathbf{p}_1^{2k} = (\mathbf{U}\mathbf{V})^{k-1}\mathbf{U}\mathbf{f}_2  \] and \[ \mathbf{p}_2^{2k+1} = (\mathbf{V}\mathbf{U})^k \mathbf{f}_2 \qquad \mathbf{p}_2^{2k} = (\mathbf{V}\mathbf{U})^{k-1}\mathbf{V}\mathbf{f}_1  \]
	If either $rank(\mathbf{U}\mathbf{V}-\mathbf{I})=|\Sigma|-1$ or $rank(\mathbf{V}\mathbf{U}-\mathbf{I})=|\Sigma|-1$ , then $\lim_{\ell\rightarrow\infty} p_i^{\ell} = \mu,i=1,2$.  \end{restatable}

Thus, even in general, $\frac{p^{\ell}_1+p^{\ell}_2}{2}$ can be seen as an approximation for prior $\mu$ in many scenarios. \citet{samet1998iterated} includes the above result (Proposition~\ref{prop:general}) and shows that the convergence of the expectation hierarchy is a necessary and sufficient condition of the existence of the common prior. Here we still present the proof of Proposition~\ref{prop:general} for completeness.   

Regarding the proof, once we realize that $\mathbf{p}^{\ell}_1 = \mathbf{U} \mathbf{p}^{\ell-1}_2$ and $\mathbf{p}^{\ell}_2 = \mathbf{V} \mathbf{p}^{\ell-1}_1$, the above result follows directly. We defer the formal proof to \Cref{sec:proofs}.

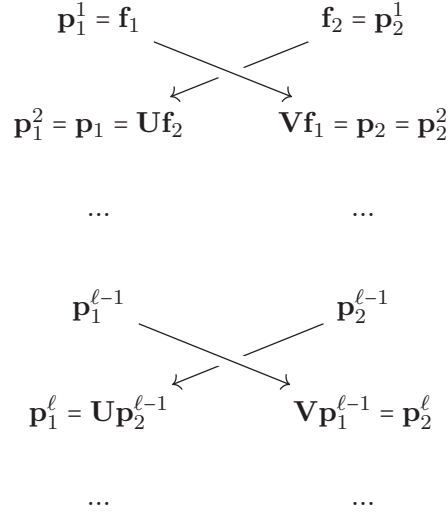
\begin{figure}
	\begin{center}
\begin{tikzcd}

%[column sep = large, row sep = large]

  \mathbf{p}_1^{1}=\mathbf{f}_1 &\mathbf{f}_2=\mathbf{p}_2^{1}  \\
  \mathbf{p}_1^{2}=\mathbf{p}_1 =\mathbf{U} \mathbf{f}_2 \ar[from=1-2,crossing over] &\mathbf{V} \mathbf{f}_1=\mathbf{p}_2=\mathbf{p}_2^{2}
  \ar[from=1-1,crossing over]\\
 ... & ...\\
  \mathbf{p}_1^{\ell-1} & \mathbf{p}_2^{\ell-1} \\
  \mathbf{p}_1^{\ell}= \mathbf{U} \mathbf{p}_2^{\ell-1} \ar[from=4-2,crossing over] & \mathbf{V} \mathbf{p}_1^{\ell-1}=\mathbf{p}_2^{\ell} \ar[from=4-1,crossing over]\\
  ... & ...\\
\end{tikzcd}
	\end{center} 
	\caption{Higher Level of Expectations}\label{fig:higherlevel}
\end{figure}

\paragraph{Limitation} Even with infinite level of expectations, we may not obtain optimal forecast in some scenarios. When two agents' signals are independent bits where each bit is 0 with probability 0.5 and $W$ is their xor, the forecasts and all levels of expectations are 1/2. No scheme can aggregate optimally based on the forecasts and expectations. Moreover, the level of expectations only depend on $\mathbf{f}_i,\mathbf{U},\mathbf{V}$. Thus, it only depends on the marginal joint distribution between each forecaster's private signal and outcome $W$, and the joint distribution between the forecasters' private signals. However, the optimal aggregation of agents' private signals require the knowledge of the joint distribution over $S_1,S_2$ and $W$. Thus, the forecasts and all levels of expectations do not provide sufficient information for perfect aggregation.

\section{Additional Proofs}\label{sec:proofs}

\claimfix*

\begin{proof}[Proof of ~\Cref{claim:fix}]
In the conditionally independent setting, 

When agent 1 receives $s_1$, her expectation for agent 2's forecast is 
	 
	 \begin{align*}
	 	p_1=&\sum_x \Pr_Q[S_2=x|S_1=s_1]\Pr_Q[W=1|S_2=x]\\ \tag{conditionally independent}
	 	=&\sum_x (\Pr_Q[W=0|S_1=s_1] \Pr_Q[S_2=x|W=0]+\Pr_Q[W=1|S_1=s_1] \Pr_Q[S_2=x|W=1])f_2^x\\ 
	 	= & \sum_x( (1-f_1^{s_1}) \Pr_Q[S_2=x|W=0] + f_1^{s_1} \Pr_Q[S_2=x|W=1]) f_2^x\\
	 	= & (1-f_1^{s_1}) z_2 + f_1^{s_1} o_2\\ \tag{When agent 1 receives $s_1$, $f_1=f_1^{s_1}$}
	 	= & (1-f_1) z_2 + f_1 o_2
	 \end{align*}
	 
Analogously, agent 2's expectation for agent 1's forecast is $p_2 = f_2 o_1 + (1-f_2) z_1$. 

Moreover, 

\begin{align*}
	&\mu o_i + (1-\mu) z_i \\
	= & \Pr_Q[W=1] \sum_x f_i^x \Pr_Q[S_i=x|W=1] + \Pr_Q[W=0] \sum_x f_i^x \Pr_Q[S_i=x|W=0]\\
	= & \sum_x f_i^x \Pr_Q[S_i=x] \\
	= & \sum_x \Pr_Q[W=1|S_i=x] \Pr_Q[S_i=x] = \mu
\end{align*}

\begin{align*}
	f^*=&\Pr_Q[W=1|S_1=s_1,S_2=s_2]\\
	=&\frac{\Pr_Q[W=1|S_1=s_1,S_2=s_2]}{\Pr_Q[W=1|S_1=s_1,S_2=s_2]+\Pr_Q[W=0|S_1=s_1,S_2=s_2]}\\
	= & \frac{\Pr_Q[W=1,S_1=s_1,S_2=s_2]}{\Pr_Q[W=1,S_1=s_1,S_2=s_2]+\Pr_Q[W=0,S_1=s_1,S_2=s_2]}\\
	= & \frac{\mu \frac{q_1^{s_1}f_1^{s_1}}{\mu}\frac{q_2^{s_2}f_2^{s_2}}{\mu}}{\mu \frac{q_1^{s_1}f_1^{s_1}}{\mu}\frac{q_2^{s_2}f_2^{s_2}}{\mu}+(1-\mu) \frac{q_1^{s_1}(1-f_1^{s_1})}{1-\mu}\frac{q_2^{s_2}(1-f_2^{s_2})}{1-\mu}}\\
	= & \frac{(1-\mu) f_1 f_2}{(1-\mu)f_1 f_2 + \mu(1-f_1)(1-f_2)}
\end{align*}

\end{proof}

\claimmultilinear*
\begin{proof}[Proof of \Cref{claim:multilinear}]
	The regret regarding $Q$ is
\begin{align*}
	& \E_{Q} (a(f_1,f_2,p_1,p_2)-f^*)^2\\
	=& \sum_{s_1,s_2} \Pr_Q[S_1=s_1,S_2=s_2] \\
	&(a(f_1^{s_1},f_2^{s_2},(1-f_1^{s_1}) z_2 + f_1^{s_1} o_2,f_2^{s_2} o_1 + (1-f_2^{s_2}) z_1)-\frac{(1-\mu) f_1^{s_1} f_2^{s_2}}{(1-\mu)f_1^{s_1} f_2^{s_2} + \mu(1-f_1^{s_1})(1-f_2^{s_2})})^2 \end{align*}

Moreover, we have $z_i=\frac{\mu-\mu o_i}{1-\mu},i=1,2$ and $ \Pr_Q[S_1=s_1,S_2=s_2]=\mu \frac{q_1^{s_1}f_1^{s_1}}{\mu}\frac{q_2^{s_2}f_2^{s_2}}{\mu}+(1-\mu) \frac{q_1^{s_1}(1-f_1^{s_1})}{1-\mu}\frac{q_2^{s_2}(1-f_2^{s_2})}{1-\mu}$.

Thus, we have $\E_{Q} (a(f_1,f_2,p_1,p_2)-f^*)^2 = \sum_{s_1,s_2} q_1^{s_1} q_2^{s_2} \psi(f_1^{s_1},f_2^{s_2},\mu,o_1,o_2)$

\end{proof}

\claimlinear*

\begin{proof}[Proof of \Cref{claim:linear}]

We pick any $\tilde{Q}\in \mathcal{Q}_{\mu,o_1,o_2,\mathbf{f}_1,\mathbf{f}_2}$. 

Constraint 1) and the non-negative constraint are satisfied naturally. $ \mu = \Pr_{\tilde{Q}}[W=1]=\sum_{x} \Pr_{\tilde{Q}}[W=1|S_i=x] \Pr[S_i=x] =\sum_x f_i^x \tilde{q}_i^x$. Thus, constraint 2) is satisfied. 

We then connect the prior vector $\mathbf{\tilde{q}}_i$ with $o_i$. $o_i = \sum_x f_i^x \Pr_{\tilde{Q}}[S_i=x|W=1]=\sum_x f_i^x \frac{f_i^x \tilde{q}_i^x }{\mu}$. Thus, constraint 3) is also satisfied.

%We then connect the prior vector $\mathbf{\tilde{q}}_i$ with $z_i,o_i$. $z_i=\sum_x f_i^x \Pr_{\tilde{Q}}[S_i=x|W=0]=\sum_x f_i^x \frac{(1-f_i^x) \tilde{q}_i^x}{1-\mu}$. $o_i = \sum_x f_i^x \Pr_{\tilde{Q}}[S_i=x|W=1]=\sum_x f_i^x \frac{f_i^x \tilde{q}_i^x }{\mu}$. Thus, constraint 3) is also satisfied. 

We have proved one direction. It's left to show the other direction: any pair of vectors that satisfies the constraints can be used to construct a prior in the set. 

Let $\mathbf{\hat{q}}_i,i=1,2$ be any pair of solutions. We will use them to construct $\hat{Q}\in\mathcal{Q}_{\mu,o_1,o_2,\mathbf{f}_1,\mathbf{f}_2}$. 

For all $(s_1,s_2)\in \Sigma \times \Sigma $, we set $\Pr_{\hat{Q}}(S_1=s_1,S_2=s_2,W=1)=\mu \frac{\hat{q}_1^{s_1}f_1^{s_1}}{\mu}\frac{\hat{q}_2^{s_2}f_2^{s_2}}{\mu}$, $\Pr_{\hat{Q}}(S_1=s_1,S_2=s_2,W=0)=(1-\mu) \frac{\hat{q}_1^{s_1}(1-f_1^{s_1})}{1-\mu}\frac{\hat{q}_2^{s_2}(1-f_2^{s_2})}{1-\mu}$. 

According to constraint 2), we have $\Pr_{\hat{Q}}[W=1]=\sum_{s_1,s_2} \Pr_{\hat{Q}}(S_1=s_1,S_2=s_2,W=1)=\sum_{s_1,s_2} \mu \frac{\hat{q}_1^{s_1}f_1^{s_1}}{\mu}\frac{\hat{q}_2^{s_2}f_2^{s_2}}{\mu}=\mu$. According to constraint 1) and 2),  $\sum_x (1-f_i^x) \hat{q}_i^x=\sum_x  \hat{q}_i^x-\sum_x f_i^x \hat{q}_i^x=1-\mu$. Thus, $\Pr_{\hat{Q}}[W=0]=\sum_{s_1,s_2} \Pr_{\tilde{Q}}(S_1=s_1,S_2=s_2,W=0)=\sum_{s_1,s_2} (1-\mu) \frac{\tilde{q}_1^{s_1}(1-f_1^{s_1})}{1-\mu}\frac{\hat{q}_2^{s_2}(1-f_2^{s_2})}{1-\mu}=1-\mu$.

The above results also imply that we have defined a valid distribution where $\sum_{s_1,s_2}\Pr_{\hat{Q}}(S_1=s_1,S_2=s_2,W=1)+ \Pr_{\hat{Q}}(S_1=s_1,S_2=s_2,W=0)=1$. 

According to constraint 2), $\Pr_{\hat{Q}}(S_2=s_2,W=1)=\sum_{s_1} \Pr_{\hat{Q}}(S_1=s_1,S_2=s_2,W=1)=\hat{q}_2^{s_2}f_2^{s_2}$. Thus, $\Pr_{\hat{Q}}(W=1|S_2=s_2)=f_2^{s_2}$. Analogously, $\Pr_{\hat{Q}}(W=1|S_1=s_1)=f_1^{s_1}$.

$\Pr_{\hat{Q}}(S_2=s_2,W=0)=\sum_{s_1} \Pr_{\hat{Q}}(S_1=s_1,S_2=s_2,W=0)=\hat{q}_2^{s_2}(1-f_2^{s_2})$. Analogously, $\Pr_{\hat{Q}}(S_1=s_1,W=0)=\hat{q}_1^{s_1}(1-f_1^{s_1})$.

Additionally, $\Pr_{\hat{Q}}(S_2=s_2)=\Pr_{\hat{Q}}(S_2=s_2,W=1)+\Pr_{\hat{Q}}(S_2=s_2,W=0)=\hat{q}_2^{s_2}$. Analogously, $\Pr_{\hat{Q}}(S_1=s_1)=\hat{q}_1^{s_1}$.

Therefore, $\hat{Q}$ defines a valid joint distribution, whose prior vectors are $\mathbf{\hat{q}}_i,i=1,2$.  $\hat{Q}$ also has the same $\mu,\mathbf{f}_1,\mathbf{f}_2$ as $Q$. 

Finally, with the above results, constraint 3) show that $\hat{Q}$ has the same $o_1,o_2$ as $Q$. Thus, we can use $\mathbf{\hat{q}}_i,i=1,2$ to construct a valid prior $\hat{Q}$ which belongs to $ \mathcal{Q}_{\mu,o_1,o_2,\mathbf{f}_1,\mathbf{f}_2}$.

\end{proof}

\lemmamix*

\begin{proof}[Proof of \Cref{lemma:mix}]

\begin{align}
	\notag &x \E_{Q} \left(\Pr_{x Q + (1-x) Q'}[W=1|f_1,f_2,p_1,p_2]-\Pr_{Q}[W=1|f_1,f_2,p_1,p_2]\right )^2 \\ \notag 
		 & + (1-x) \E_{Q'} \left(\Pr_{x Q + (1-x) Q'}[W=1|f_1,f_2,p_1,p_2]-\Pr_{Q'}[W=1|f_1,f_2,p_1,p_2]\right)^2\\ \notag
	\geq & x\Pr_Q [S_1 = B, S_2=B] ( \Pr_{x Q + (1-x) Q'}[W=1|f_1=f_2=b,p_1=p_2=p]-\frac{(1-\mu) b^2 }{(1-\mu)b^2 + \mu(1-b)(1-b)})^2
	\\&+(1-x)\Pr_{Q'} [S_1 = B, S_2=B] (\Pr_{x Q + (1-x) Q'}[W=1|f_1=f_2=b,p_1=p_2=p]-\frac{(1-\mu') b^2 }{(1-\mu')b^2 + \mu'(1-b)(1-b)})^2 \label{eq:lbf}
\end{align}

Moreover, the Bayesian posterior conditioning $f_1=f_2=b,p_1=p_2=p$ is the posteriors conditioning on both agents receive signal B. 

	\begin{align*}
	&\Pr_{x Q + (1-x) Q'}[W=1|f_1=f_2=b,p_1=p_2=p]\\
	= & \Pr_{x Q + (1-x) Q'}[W=1|S_1=B,S_2=B]\\
	= & \frac{\Pr_{x Q + (1-x) Q'}[W=1,S_1=B,S_2=B]}{\Pr_{x Q + (1-x) Q'}[S_1=B,S_2=B]}\\
	= & \frac{\Pr_{x Q + (1-x) Q'}[W=1,S_1=B,S_2=B]}{\Pr_{x Q + (1-x) Q'}[S_1=B,S_2=B]}\\ \tag{$q^B=\Pr_Q[S_1=B],q'^B=\Pr_{Q'}[S_1=B]$}
	= & \frac{x \Pr_Q [W=1,S_1 = B, S_2=B] + (1-x)\Pr_{Q'} [W=1,S_1 = B, S_2=B]}{x \Pr_Q [S_1 = B, S_2=B] + (1-x)\Pr_{Q'} [S_1 = B, S_2=B]}
\end{align*} 

We have $\Pr_Q [W=1,S_1 = B, S_2=B]= \mu (\frac{b q^B }{\mu})^2$, and $\Pr_{Q'} [W=1,S_1 = B, S_2=B]=\mu'(\frac{b q'^{B}}{\mu'})^2$. Additionally,

\begin{align*}
	\Pr_Q [S_1 = B, S_2=B] = \mu (\frac{b q^B }{\mu})^2 + (1-\mu) (\frac{(1-b) q^B }{1-\mu})^2
\end{align*}

Analogously, 

\begin{align*}
	\Pr_{Q'} [S_1 = B, S_2=B] = \mu' (\frac{b q'^B }{\mu'})^2 + (1-\mu') (\frac{(1-b) q'^B }{1-\mu'})^2
\end{align*}

It's left to represent $q^B$ and $q'^B$ by $\mu,\mu',b,p,x$. 

According to \Cref{claim:fix} and \Cref{claim:linear}, $p=b (b^2 \frac{q^B}{\mu}+\frac{q^C}{\mu})+(1-b) (b(1-b)\frac{q^B}{1-\mu})$. 

Additionally, $\mu = b q^B + q^C$ which implies that $q^C=\mu- b q^B$.

Thus, we have $q^B (\frac{b(1-b)^2}{1-\mu}-\frac{b^2-b^3}{\mu}) =p-b$, which implies that $q^B = \frac{p-b}{\frac{b(1-b)^2}{1-\mu}-\frac{b^2-b^3}{\mu}}$. 

Analogously, $q'^B = \frac{p-b}{\frac{b(1-b)^2}{1-\mu'}-\frac{b^2-b^3}{\mu'}}$.

Finally, we can substitute $q^B$ and $q'^B$. By symbolic computations aided by Matlab, we have \Cref{eq:lbf}$= (p-b)^2 \psi(\mu,\mu',b,x)$.

Thus, to have a better lower bound, we want $(p-b)^2$ to be large. However, because $q^A,q^B,q^C$, $q'^A,q'^B,q'^C\in[0,1]$, we also have restrictions on $p$. 

We consider the case $b<\mu<\mu'$. 

%In this case, $\frac{b(1-b)^2}{1-\mu'}-\frac{b^2-b^3}{\mu'}=b(1-b)(\frac{1-b}{1-\mu'}-\frac{b}{\mu'})\geq b(1-b)(\frac{1-b}{1-\mu}-\frac{b}{\mu})> 0$. 

%1) $q^B,q'^B\in [0,1]\Rightarrow$ $0\leq p-b\leq b(1-b)(\frac{1-b}{1-\mu}-\frac{b}{\mu})$, $0\leq p-b\leq b(1-b)(\frac{1-b}{1-\mu'}-\frac{b}{\mu'})$

1) $q^A,q'^A\in [0,1]\Rightarrow$ $q^B+q^C,q'^B+q'^C\in [0,1]$ $\Rightarrow$ $(1-b) q^B\leq 1-\mu $, $(1-b) q'^B\leq 1-\mu' $. This implies that $p-b \leq \frac{1-\mu}{1-b} b(1-b)(\frac{1-b}{1-\mu}-\frac{b}{\mu})=b(1-b-b \frac{1-\mu}{\mu})=b(1-b\frac{1}{\mu})$ and  $p-b \leq b(1-b\frac{1}{\mu'})$. 

2) $q^B,q'^B\in [0,1]$ are satisfied directly with restriction 1) because both $\frac{1-\mu}{1-b}$ and $\frac{1-\mu'}{1-b}$ are less than one in the case $b<\mu<\mu'$. 

3) $q^C,q'^C\in [0,1]\Rightarrow$ $b q^B\leq \mu$, $b q'^B\leq \mu'$. Because we are in the case $b<\mu<\mu'$, the above restrictions are satisfied automatically. 

Therefore, in the case $b<\mu<\mu'$, putting the above restrictions together, $p-b\leq b(1-b\frac{1}{\mu})$.

To this end, we can substitute the optimal $p$ into \Cref{eq:lbf} and it's left to maximize

\[(b(1-b\frac{1}{\mu}))^2 \psi (\mu,\mu',b,x)\] with $b<\mu<\mu'$.

We then fix $\mu,\mu',b$ and substitute $x$ by the solution of $\frac{\partial\psi}{\partial x}(x)=0 $. Finally, we reduce the parameters to only $\mu,\mu',b$ and find that $\mu = 0.7427, \mu' = 0.9100, b=0.62$ leads to regret $>0.00144$.

% and corresponds to the following $Q,Q'$:

%\begin{center}
%\begin{tabular}{ |c| c c |c| c c |}
%\hline
% $Q\approx$  & 1 & 0 & $Q'\approx$ & 1 & 0\\ 
%  \hline 
% A  & 0 & 0 & A & 0 &  0.0043\\  
% B & 0.4610 $\times$ 0.86 & 0.4610 $\times$ 0.14 & B & 0.0496 $\times$ 0.86 & 0.0496 $\times$ 0.14  \\
% C  & 0.5390 & 0 & C &  0.9461 & 0\\
% \hline 
%\end{tabular}
%\end{center}

\end{proof}

\claimiid*

\begin{proof}[Proof of \Cref{claim:iid}]
In the c.i.i.d. setting, $o_1=o_2$ and $z_1=z_2$. Let $o$ denote $o_1$ and $z$ denote $z_1$. 

\Cref{claim:fix} shows that $p_1 = f_1 o + (1-f_1) z$ and $p_2 = f_2 o + (1-f_2) z$. Therefore, 
$f_2 p_1 - f_1 p_2 = (f_2-f_1) z$. 

Moreover, \Cref{claim:fix} also shows that $\mu o + (1-\mu)z=\mu$. If $\mu = 1$, then $f_1=f_2=1$, which contradicts our condition $f_1\neq f_2$. Thus, $\mu\neq 1$. We can represent $z = \frac{\mu}{1-\mu} (1-o)$. 

Finally, 

\begin{align*}
	\frac{f_2 p_1 -f_1 p_2}{f_2-f_1 -p_2 + p_1} = & \frac{(f_2-f_1) z}{f_2-f_1 + o(f_1-f_2) +(f_2-f_1)z}\\
	= & \frac{z}{1-o+z}\\ \tag{$z = \frac{\mu}{1-\mu} (1-o)$}
	= & \mu 
\end{align*}

\end{proof}

\prophigh*

\begin{proof}[Proof of \Cref{prop:high}]

For all random variable $X$, let $\E_i[X]$ denote the expectation for $X$ from agent $i$'s perspective. Let $\E_{W=0}[X]=\sum_x x\Pr_Q[X=x|W=0]$ and $\E_{W=1}[X]=\sum_x x\Pr_Q[X=x|W=1]$.

In the conditionally independent setting, we have $p_1^{\ell} = \E_1[p_2^{\ell-1}]=f_1 \E_{W=1}[p_2^{\ell-1}]+ (1-f_1) \E_{W=0}[p_2^{\ell-1}]$ and $p_2^{\ell} = \E_2[p_1^{\ell-1}]=f_2 \E_{W=1}[p_1^{\ell-1}]+ (1-f_2) \E_{W=0}[p_1^{\ell-1}]$. 

Let $o_i^{\ell} = \E_{W=1}[p_i^{\ell}] $ and $z_i^{\ell} = \E_{W=0}[p_i^{\ell}] $.

$o_2^{\ell}=\E_{W=1}[p_2^{\ell}]=\sum_s \Pr[S_2=s|W=1](f_2^s \E_{W=1}[p_1^{\ell-1}]+(1-f_2^s))\E_{W=0}[p_1^{\ell-1}])$, where $\Pr[S_2=s|W=1]=\frac{q_2^s f_2^s}{\mu}$.

$z_2^{\ell}=\E_{W=0}[p_2^{\ell}]=\sum_s \Pr[S_2=s|W=0](f_2^s \E_{W=1}[p_1^{\ell-1}]+(1-f_2^s))\E_{W=0}[p_1^{\ell-1}])$, where $\Pr[S_2=s|W=0]=\frac{q_2^s (1-f_2^s)}{1-\mu}$.

Thus, we have $o_2^{\ell} = \sum_s \frac{q_2^s f_2^s}{\mu} (f_2^s o_1^{\ell-1}+(1-f_2^s))z_1^{\ell-1})$ and $z_2^{\ell} = \sum_s \frac{q_2^s (1-f_2^s)}{1-\mu} (f_2^s o_1^{\ell-1}+(1-f_2^s))z_1^{\ell-1})$.

Analogously, we have

$o_1^{\ell} = \sum_s \frac{q_1^s f_1^s}{\mu} (f_1^s o_2^{\ell-1}+(1-f_1^s))z_2^{\ell-1})$ and

$z_1^{\ell} = \sum_s \frac{q_1^s (1-f_1^s)}{1-\mu} (f_1^s o_2^{\ell-1}+(1-f_1^s))z_2^{\ell-1})$.

Additionally, based on induction, $\mu o_i^{\ell} + (1-\mu) z_i^{\ell}=\E[p_i^{\ell}]=\mu$ for $i=1,2$. 

When $o^{\ell-1}_1=z^{\ell-1}_1=\mu$, $o^{\ell}_2=\sum_s \frac{q_2^s f_2^s}{\mu}\mu = \mu$ and $z^{\ell}_2 = \sum_s \frac{q_2^s (1-f^s_2)}{1-\mu}\mu =\mu $. Analogously, when $o^{\ell}_2=z^{\ell}_2=\mu$, $o^{\ell+1}_1=z^{\ell+1}_1=\mu$. Therefore, $o^{\ell-1}_i=z^{\ell-1}_i=\mu,i=1,2$ is a fixed point. 

It's left to show the convergence of the series $\{o_i^{\ell}, z_i^{\ell}\},i=1,2$. Without loss of generality, we assume $q_i^s>0,\forall s$, because otherwise we can restrict $\Sigma$ to the set of signals which have non-zero prior probability.

Because $\sum_s \frac{q_1^s f_1^s}{\mu}=1$, $o_1^{\ell}$ is a convex combination of $\{(f_1^s o_2^{\ell-1}+(1-f_1^s))z_2^{\ell-1})\}_s$. Thus, we have $o_1^{\ell}\leq\max\{o^{\ell-1}_2,z^{\ell-1}_2\}$, and analogously, $z_1^{\ell} \geq \min\{o^{\ell-1}_2,z^{\ell-1}_2\}$. 

The equality holds only if 1) $o^{\ell-1}_2=z^{\ell-1}_2$ or 2) $\forall s$, $f^s_1=0$ or $f^s_1=1$. In case 1), $\mu = \mu o^{\ell-1}_2 + (1-\mu) z^{\ell-1}_2=o^{\ell-1}_2$. The series $\{o^{\ell}_i, z^{\ell}_i\}, i=1,2$ have already converged to $\mu$ according to the above observation. In case 2), agent 1 is a perfect forecaster which contradicts the condition. 

Therefore, when the series have not converged, the equality does not hold, $|o_1^{\ell}-z_1^{\ell}|<|o^{\ell-1}_2-z^{\ell-1}_2|$ and analogously,  $|o^{\ell-1}_2-z^{\ell-1}_2|<|o^{\ell-2}_1-z^{\ell-2}_1|$.

Therefore, as $\ell$ goes to infinite, $\lim_{\ell\rightarrow\infty} o^{\ell}_i-z^{\ell}_i=0$, which implies that $\lim_{\ell\rightarrow\infty} o^{\ell}_i=\lim_{\ell\rightarrow\infty} z^{\ell}_i= \mu,i=1,2$. Therefore, $\lim_{\ell\rightarrow\infty} p_i^{\ell} = f \lim_{\ell\rightarrow\infty} o^{\ell}_i + (1-f) \lim_{\ell\rightarrow\infty} z^{\ell}_i = \mu$. 

Thus \[\lim_{\ell\rightarrow\infty} L(a_{ae}^{\ell},Q)=0.\]

\end{proof}

\propgeneral*

\begin{proof}[Proof of \Cref{prop:general}]
	According to the definition of level $\ell$ expectation, $\mathbf{p}^{\ell}_1(s_1) = \sum_s \Pr[S_2=s|S_1=s_1] \mathbf{p}^{\ell-1}_2(s)$. Therefore, $\mathbf{p}^{\ell}_1 = \mathbf{U} \mathbf{p}^{\ell-1}_2$. Analogously, $\mathbf{p}^{\ell}_2 = \mathbf{V} \mathbf{p}^{\ell-1}_1$.
	
	Note that $\mathbf{p}^{1}_i =  \mathbf{f}_i, i=1,2$. Thus, \[ \mathbf{p}_1^{2k+1} = (\mathbf{U}\mathbf{V})^k \mathbf{f}_1 \qquad \mathbf{p}_1^{2k} = (\mathbf{U}\mathbf{V})^{k-1}\mathbf{U}\mathbf{f}_2  \] and \[ \mathbf{p}_2^{2k+1} = (\mathbf{V}\mathbf{U})^k \mathbf{f}_2 \qquad \mathbf{p}_2^{2k} = (\mathbf{V}\mathbf{U})^{k-1}\mathbf{V}\mathbf{f}_1  .\]
	
	Because both $\mathbf{U}$ and $\mathbf{V}$ are row-stochastic, both $\mathbf{U}\mathbf{V}$ and $\mathbf{V}\mathbf{U}$ are row-stochastic as well. 
	
	Thus, $(\mathbf{U}\mathbf{V}-\mathbf{I})\mathbf{v}=0$ has at least one solution, $\mathbf{v}=\mathbf{1}$ which is a column vector whose elements are all one. 
	
	Without loss of generality, we assume that $rank(\mathbf{U}\mathbf{V}-\mathbf{I})=|\Sigma|-1$. This implies that its kernel space's dimension is 1. 
	
	Therefore, $\mathbf{p}^{\ell}_1$ will converge to a vector $\mathbf{p}^*$ that is proportional to $\mathbf{1}$. 
	
	Additionally, let $\mathbf{q}_1$ denote a column vector where $\mathbf{q}_1(s)=q^s_1=\Pr_Q[S_1=s]$. 
	
	$(\mathbf{q}_1^{\top} \mathbf{U}\mathbf{V})(s'') = \sum_{s,s'}  \Pr_Q[S_1=s] \Pr_Q[S_2=s'|S_1=s] \Pr_Q[S_1=s''|S_2=s'] =\Pr_Q[S_1=s'']$. Therefore, $\mathbf{q}_1^{\top} \mathbf{U}\mathbf{V} = \mathbf{q}_1$, which implies that $\mathbf{q}_1^{\top} \mathbf{p}^{2k}_1 = \mathbf{q}_1^{\top} \mathbf{f}_1 = \mu $ for all $k$. 
	
	Thus, $\mathbf{q}_1^{\top}\mathbf{p}^* = \mu$. Combining with the fact that $\mathbf{p}^*$ is proportional to $\mathbf{1}$, we have $\mathbf{p}^*= \mu \mathbf{1}$. 
	
	When $\mathbf{p}^{\ell}_1$ converges to $\mu \mathbf{1}$, $\mathbf{p}^{\ell}_2$ converges to $\mu \mathbf{1}$ as well based on the recursive relationship. 
	
	Thus, $\lim_{\ell\rightarrow\infty} p_i^{\ell} = \mu,i=1,2$.

\end{proof}

\propmany*

\begin{proof}[Proof of \Cref{prop:many}]
When there exists $f_i\neq f_j$, \Cref{claim:iid} implies that $\hat{\mu}=\mu$, thus the regret is zero. It's left to analyze the scenarios when all agents report the same forecast.

        \begin{align*}
		L(a^n_{ciid},Q) = & \sum_s \Pr[\forall i, S_i=s] \\
		&\times (\frac{(1-\hat{\mu})^{n-1}(f^s)^n}{(1-\hat{\mu})^{n-1}(f^s)^n+\hat{\mu}^{n-1}(1-f^s)^n}\\
		&\ \ \ \ \ \ -\frac{(1-\mu)^{n-1}(f^s)^n}{(1-\mu)^{n-1}(f^s)^n+\mu^{n-1}(1-f^s)^n})^2\\
		\end{align*} where $\hat{\mu} = f^s o + (1-f^s) z$.

		If $f^s=0$, then the regret is zero. Otherwise, we rewrite the regret into the following format
		
		\begin{align*}
			= & \sum_s \Pr[\forall i, S_i=s] \\
		&\times (\frac{1}{(1+\frac{1-\hat{\mu}}{\hat{\mu}}(\frac{\hat{\mu}(1-f^s)}{(1-\hat{\mu})f^s})^n}-\frac{1}{(1+\frac{1-\mu}{\mu}(\frac{\mu(1-f^s)}{(1-\mu)f^s})^n})^2\\\end{align*}

\begin{claim}\label{claim:many}
	$ f^s- \hat{\mu} = \frac{1-o}{1-\mu}(f^s-\mu)$.
\end{claim}
\begin{proof}[Proof of \Cref{claim:many}]
\begin{align*}
	 f^s- \hat{\mu} = & f^s-f^s o -(1-f^s)\frac{\mu-\mu o}{1-\mu}\\
	 = & (1-o)(f^s-(1-f^s) \frac{\mu}{1-\mu}) \\
	 = & (1-o)(\frac{f^s(1-\mu)}{1-\mu}-(1-f^s) \frac{\mu}{1-\mu})\\
	 = & (1-o)\frac{1}{1-\mu} (f^s-\mu)
\end{align*}
\end{proof}

Fixing $Q$, if $f^s=\mu$, then $\hat{\mu}=\mu o + (1-\mu) z=\mu$. This implies $L(a^n_{ciid},Q)=0$. If $f^s>\mu$, \Cref{claim:many} implies $f^s>\hat{\mu}$. Then both $\frac{1-\mu}{\mu}(\frac{\mu(1-f^s)}{(1-\mu)f^s})^n$ and $\frac{1-\hat{\mu}}{\hat{\mu}}(\frac{\hat{\mu}(1-f^s)}{(1-\hat{\mu})f^s})^n$ converge to 0, which implies that $\lim_{n\rightarrow \infty} L(a^n_{ciid},Q)=0$. If $f<\mu$, \Cref{claim:many} implies $f<\hat{\mu}$. Then both $\frac{1-\mu}{\mu}(\frac{\mu(1-f^s)}{(1-\mu)f^s})^n$ and $\frac{1-\hat{\mu}}{\hat{\mu}}(\frac{\hat{\mu}(1-f^s)}{(1-\hat{\mu})f^s})^n$ go to infinite, which implies that $\lim_{n\rightarrow \infty} L(a^n_{ciid},Q)=0$ as well.

\end{proof}

\section{Strategic Agents}\label{sec:strategic}

 We provide a payment scheme that incentivizes truthful reports for both the forecasts and expectations. 

\paragraph{Payment Scheme} Given the outcome is $w$, agent $1$ will be paid $c-(w-f_1)^2-(p_1-f_2)^2$. Agent $2$ will be paid $c-(w-f_2)^2-(p_2-f_1)^2$ where $c>0$ is a constant. 

We can use any proper scoring rule to design the payment scheme. In the above scheme we pick quadratic scoring rule.

An agent's strategy is a mapping $\mathrm{s}$ from her truthful pair of forecast and expectation $(f,p)$ to a distribution over $[0,1]\times [0,1]$. When she performs the strategy, she will draw a random pair of forecast and expectation based on the distribution $\mathrm{s}(f,p)$. 

A strategy profile consists of all agents' strategies. Truth-telling defines a strategy profile that both agents tell the truth. A strategy profile and agents' beliefs consist of a perfect Bayesian equilibrium if (1) each agent’s strategy leads to optimal actions, given her beliefs and the strategies of the other agents, and (2) the beliefs are consistent. 

\begin{proposition}\label{prop:pay}
	Truth-telling is the only perfect Bayesian equilibrium in the above payment scheme. 
\end{proposition}

\begin{proof}[Proof of \Cref{prop:pay}]
	Based on the property of quadratic scoring rule, it's strictly optimal for agent 1 and 2 to report their forecasts truthfully. Conditioning on that, it's also strictly optimal for agent 1 and 2 to report their expectations truthfully.
\end{proof}

\section{Many C.I.I.D. Conditionally Independent Agents: The Original Setting}\label{sec:manyregret}

\Cref{prop:manyregret} is implicitly implied by \citet{arieli2018robust}.

\begin{proposition}\label{prop:manyregret}
	For all aggregator $a$ with only agents' forecasts, there exists $Q$ such that when the number of agents is infinite, $L(a,Q)\geq \frac{3}{32}$. 
\end{proposition}

\begin{proof}[Proof of \Cref{prop:manyregret}]
	\citet{arieli2018robust} propose the following two information structures. Nature picks $Q_1$ with probability $\frac{3}{8}$ and $Q_2$ with probability $\frac{5}{8}$.

\begin{center} 
\begin{tabular}{ |c| c c |c| c c |}
\hline
 $Q_1\approx$  & 1 & 0 & $Q_2\approx$ & 1 & 0\\ 
  \hline 
 A  & $\frac{1}{8}$ & $\frac{3}{8}$ & A & $\frac{1}{40}$ &  $\frac{3}{40}$ \\ 
 B & $\frac{3}{8}$  & $\frac{1}{8}$ & B & $\frac{27}{40}$ &  $\frac{9}{40}$ \\
 \hline  
\end{tabular}
\end{center}
 
Even when the number of agents is infinite, when $1/4$ fraction of agents forecast $\frac14$ and, $3/4$ fraction of agents forecast $\frac34$, the Bayesian aggregator cannot infer nature's choice. Let $E$ denote this event. By calculations, the Bayesian aggregator's optimal forecast is $\frac12$. 

Thus, the regret of the Bayesian aggregator is at least

\begin{align*} & \frac{3}{8} \Pr_{Q_1}[E] (\frac{1}{2} - \Pr_{Q_1}[W=1|E])^2 + \frac{5}{8} \Pr_{Q_2}[E] (\frac{1}{2} - \Pr_{Q_2}[W=1|E])^2\\
	= & \frac{3}{8} \Pr_{Q_1}[W=1]*(\frac{1}{2} - 1)^2 + \frac{5}{8} \Pr_{Q_2}[W=0](\frac12 - 0)^2\\
	= & \frac{3}{8} \frac{1}{2} (\frac{1}{2} - 1)^2 + \frac{5}{8} \frac{12}{40}(\frac12 - 0)^2\\
	= & \frac{3}{32}\approx 0.09375 \end{align*}

According to \Cref{obs:mix}, for all aggregator $a$, even when the number of agents is infinite, either $L(a,Q_1)$ or $L(a,Q_2)$ will be $\geq \frac{3}{32}\approx 0.09375$. 

\end{proof}

\section{Numerical Experiment Software}

We use Matlab (2022a) global optimization toolbox to estimate the maximum. We first generate the symbolic formulas of the corresponding reduced functions. We then run the GlobalSearch solver multiple times. For cases without symbolic formulas, we run other solvers. For function with 5 variables, we also run the grid search with 0.01 as the step-size. We pick the estimation result as the maximum over all runs. The estimation for Simple Average and Average-Prior scheme match the numerical results in \citet{arieli2018robust}.

\end{document}